\documentclass[prl,aps,twocolumn,showpacs,floatfix,amsmath,amssymb,aps]{revtex4-2}
\usepackage{amsfonts}
\usepackage{amssymb}
\usepackage[linesnumbered,ruled,vlined]{algorithm2e}
\usepackage{hyperref}
\usepackage{tikz-cd}
\usepackage{graphicx}
\usepackage{dcolumn}
\usepackage{bm,amsmath,verbatim}
\usepackage{mathrsfs}
\usepackage{color}
\usepackage{mathtools}
\usepackage{amsthm}
\newtheorem{theorem}{Theorem}[section]

\usepackage[T1]{fontenc}
\usepackage[latin9]{inputenc}
\usepackage{booktabs}
\usepackage{amstext}
\usepackage{amsopn}
\usepackage{algorithmicx}
\usepackage{algpseudocode}
\hypersetup{colorlinks,linkcolor=blue,citecolor=blue,filecolor=blue,urlcolor=blue}

\newcommand{\ket}[1]{|#1\rangle}

\newcommand{\ii}{\text{i}}

\begin{document}
\title{Low-Overhead Fermionic Fault-Tolerant Quantum Computation}
\author{Chong-Yuan Xu$^{1}$}
\thanks{These authors contribute equally to this work.}
\author{Ze-Chuan Liu$^{1}$}
\thanks{These authors contribute equally to this work.}
\author{Yong Xu$^{1,2}$}
\email{yongxuphy@tsinghua.edu.cn}
\affiliation{$^{1}$Center for Quantum Information, IIIS, Tsinghua University, Beijing 100084, People's Republic of China}
\affiliation{$^{2}$Hefei National Laboratory, Hefei 230088, People's Republic of China}

\begin{abstract}
Simulating fermionic systems on qubit-based quantum computers often demands significant computational resources due to the requirement to map fermions to qubits.
Thus, designing a fault-tolerant quantum computer that operates directly with fermions offers an effective solution to this challenge.
Here, we introduce a protocol for fault-tolerant fermionic quantum computation utilizing fermion-to-fermion low-density parity-check (LDPC) codes.
Our architecture employs a fermionic LDPC memory, which transfers its state to fermionic color code processors, where logical operations are subsequently performed.
We use odd-weight logical Majorana operators to encode fermionic information in the LDPC memory and provide an algorithm for identifying these operators. 
Explicit code examples exhibit logical failure rates substantially below the physical error rate. We further develop two fermionic lattice surgery protocols to 
enable state transfer between the memory and processors. Finally, we demonstrate effective error suppression in a simulation of fermionic dynamics using the proposed protocol.
\end{abstract}
\maketitle
\maketitle

The study of strongly correlated fermionic systems is central to high-energy physics~\cite{bauer2023quantum}, material science~\cite{bauer2020quantum}, and quantum chemistry~\cite{mcardle2020quantum}, promising insights into
phenomena ranging from quark dynamics~\cite{Berges2021rmp} to high-temperature superconductivity~\cite{varma2020colloquium}.
However, solving these problems using classical computers is often very challenging.
For instance, quantum Monte Carlo methods usually face the sign problem when addressing fermionic problems~\cite{Troyer2005prl}.
In this context, quantum computers offer a promising alternative~\cite{Feynman1982ijtp}.
However, since conventional quantum computers use qubits, solving fermionic problems requires mapping fermionic operators to qubit operators, which often incurs substantial overhead~\cite{Abrams1997prl,Ortiz2001pra,Bravyi2002annp,Ball2005prl,verstraete2005mapping,Whitfield2016pra,maskara2025fast,constantinides2025lowdepth}, making experimental implementation significantly challenging in the near term.
This motivates the development of programmable fermionic quantum processors~\cite{Bravyi2002annp,gonzalez2023fermionic}.
Recent experiments have made substantial progress toward realizing such processors~\cite{bojovic2026high,kiefer2026digital}.
To protect fermionic quantum information against unavoidable noise, 
recent proposals for fault-tolerant fermionic quantum computation employ fermion-to-fermion repetition and color codes~\cite{Schuckert2024arxiv,ott2024error}.
However, the protocol's overhead remains substantial because each code block encodes only a single logical fermion, resulting in a low encoding rate. 

Protocols based on quantum LDPC codes offer a promising route to reducing the resource overhead of qubit-based fault-tolerant quantum computation~\cite{tillich2013quantum,Gottesman2013,kovalev2013quantum,breuckmann2016constructions,panteleev2021degenerate,breuckmann2021balanced,krishna2021fault,higgott2021subsystem,breuckmann2021quantum,panteleev2021quantum,delfosse2021bounds,baspin2022connectivity,baspin2022quantifying,tremblay2022constant,cohen2022low,panteleev2022asymptotically,leverrier2022quantum,strikis2023quantum,quintavalle2023partitioning,lin2024quantum,xu2024constant,bravyi2024high,zhang2025time,li2025low}.
Compared with the surface code, these codes can achieve substantially higher encoding rates.
Although many such constructions require long-range qubit connectivity, recent advances in platforms~\cite{bluvstein2022quantum,bluvstein2024logical,bravyi2024high,chen2024benchmarking,wang2026demonstration} such as Rydberg atom arrays, superconducting qubits, and trapped ions have made the near-term implementation of these codes promising.
Despite this progress, existing studies have focused exclusively on qubit-based LDPC architectures.
It remains unclear whether one can construct a fermion-to-fermion LDPC code capable of encoding multiple fermionic modes in a single code block and how to perform
logical operations on such codes.

\begin{figure}
\includegraphics[width=\linewidth]{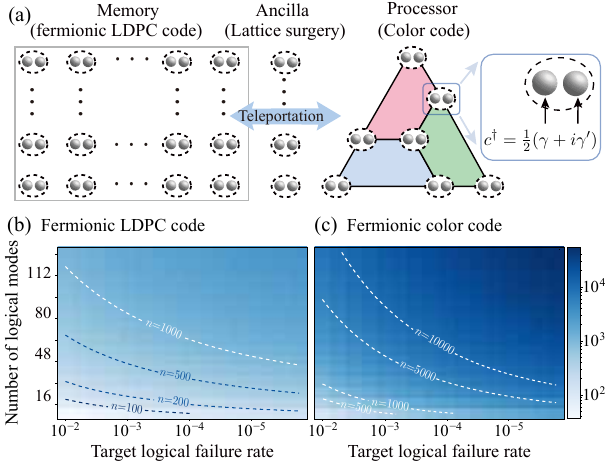}
\caption{
(a) Schematic illustration of the fault-tolerant fermionic quantum-computation framework.
(b) Number of physical fermionic sites required for a memory encoded by fermionic LDPC codes, for a given target logical failure rate and a given number of encoded logical fermionic modes. 
(c) Resource overhead of a memory encoded by fermionic color codes. The overhead is determined by the minimum value obtained from extrapolating the numerical results. 
}
\label{Architecture}
\end{figure}

Here, we introduce a protocol for fault-tolerant fermionic quantum computation based on fermion-to-fermion LDPC codes.
Our approach employs a fermionic LDPC memory integrated with fermionic color code processors as illustrated in Fig.~\ref{Architecture}, which is generalized from the quantum LDPC case~\cite{xu2024constant}.
Logical fermionic information is initially stored in the fermionic LDPC memory and, when needed, is transferred to fermionic color code processors, where fault-tolerant logical operations are conducted.
Upon completion, the information is returned to the memory.
One natural method is to consider Majorana stabilizer codes with odd-weight logical Majorana operators.
However, to the best of our knowledge, all codes satisfying this requirement have a significant drawback~\cite{cao2012novel,aly2008class,farinholt2012quantum}, as Ref.~\cite{MacKay2004ieee} stated, they have codewords of weight below stabilizer weight, resulting in significant resource overhead for syndrome extraction.
Here, we consider the stacked LDPC codes proposed by us in Ref.~\cite{liu2026stacked} motivated by this work.
These codes are characterized by logical Majorana operators consisting of an odd number of physical Majorana operators, and can have a significantly higher encoding rate compared to the fermionic color code.
Figure~\ref{Architecture}(b) and (c) show the overhead of physical fermionic modes for memory based on our LDPC code and memory based on color code, respectively (see Supplemental Material for simulation details). 
Our scheme achieves nearly an order of magnitude of resource savings.
Furthermore, to facilitate transfer of logical fermionic states between the fermionic memory and processor, fermionic lattice surgery is needed.
Fermionic lattice surgery is different from qubit lattice surgery due to their different commutativity (see Supplemental Material for more details).
We propose methods for performing fermionic lattice surgery between fermionic LDPC codes, which make fermionic teleportation work.
We show that this process maintains the code distance, and is fault-tolerant under circuit-level error model.
Finally, we simulate the dynamics of a fermionic system, demonstrating effective error suppression.

We start by introducing the construction of our fermionic LDPC codes from $2n$ physical Majorana operators $\left\{{\gamma}_1, {\gamma}_1^{\prime},{\gamma}_2, {\gamma}_2^{\prime}, \dots,{\gamma}_n, {\gamma}_n^{\prime}\right\}$ arising from $n$ physical complex fermions' creation and annihilation operators~\cite{kitaev2006anyons}.
The strings of these Majorana operators, together with a phase factor $\eta\in\left\{\pm1,\pm\ii\right\}$,
generate the group of Majorana operators $\text{Maj}(2n)\equiv\left\{{\Gamma}=\eta\prod_{j=1}^{n} {\gamma}_j^{\alpha_j} ({\gamma}_j^\prime)^{\alpha_j^\prime}|\alpha_j,\alpha_j^\prime \in \{0,1\} \right\}$~\cite{Bravyi2010njp}.
The number of Majorana operators in ${\Gamma}$ is referred to as its weight.
The Majorana stabilizer code is defined by a Majorana stabilizer group $\mathcal{S}_{\text{maj}}$, a subgroup of $\text{Maj}(2n)$, where all elements are Hermitian, mutually commute, ensuring it is an Abelian subgroup, have even weight to preserve the parity of a physical fermionic system, and $-I$ with $I$ being the identity element in the group is excluded~\cite{Bravyi2010njp}.
The logical operators are generated by an independent subset of $\text{Maj}(2n)$ comprising Majorana operator strings that commute with all elements of the stabilizer group $\mathcal{S}_{\text{maj}}$ but are not members of $\mathcal{S}_{\text{maj}}$~\cite{nielsen2010quantum}.

Each Majorana string operator can be represented by a binary vector $(\alpha_1,\alpha_1^\prime,\dots,\alpha_n, \alpha_n^\prime)$, where
$\alpha_j,\alpha_j^\prime\in \{0,1\}$ and $1\le j \le n$.
The binary vectors corresponding to the $m$ independent stabilizer generators of $\mathcal{S}_{\text{maj}}$ form an $m \times 2n$ check matrix.
For a fermion-to-fermion LDPC code, we construct a check matrix $H$ based on a weakly-self-dual CSS code with the binary check matrices that satisfy $H_X=H_Z=A$ and $AA^T=0$, as follows,
\begin{equation}\label{eq:CheckMatrix}
H=\left(\begin{matrix}
H_{\gamma} & 0\\
0 & H_{\gamma^{\prime}}
\end{matrix}\right),
\end{equation}
where $H_{\gamma}=H_{\gamma^{\prime}}=A$. $H_{\gamma}$ and $H_{\gamma^{\prime}}$ correspond to check matrices that characterize the stabilizers consisting exclusively of $\{\gamma_j\}$ and $\{\gamma_j^{\prime}\}$ operators, respectively, as described for fermion-to-fermion color codes in Ref.~\cite{Schuckert2024arxiv}.
The logical operators are represented by vectors in the kernel of $H$ that 
cannot be generated by stabilizers, whose span forms a homology group $\ker(A)/\text{im}(A^T)$ 
for both $\gamma$ and $\gamma^\prime$ types of logical operators~\cite{bombin2007homological}.

In our protocol, there are multiple logical fermionic modes within a single fermionic LDPC memory, and these logical fermionic modes are transferred to processors for the execution of logical operations.
It is essential that the logical Majorana operator in memory anticommutes with that in the processor to adhere to fermionic statistics.
We thus require all logical Majorana operators in memory have odd weight and even overlap with each other~\cite{Bravyi2010njp}.
Such conditions also ensure that the fermionic statistics are maintained when concatenating different fermionic code blocks~\cite{Schuckert2024arxiv}.
We denote the logical Majorana operators as $\overline{\gamma}_j$ 
and $\overline{\gamma}_j^\prime$ with $1 \le j \le k $ ($k$ is the 
number of logical fermionic modes), which consist of physical Majorana 
operators $\{\gamma_1,\dots,\gamma_n \}$ and $\{\gamma_1^\prime,\dots,\gamma_n^\prime \}$, respectively.
The logical complex-fermion creation operators are then defined as $\overline{c}_j^\dagger=\frac{1}{2}(\overline{\gamma}_j+\ii \overline{\gamma}_j^{\prime})$~\cite{kitaev2006anyons}.

\begin{figure}[t]
\centering
\includegraphics[width=\linewidth]{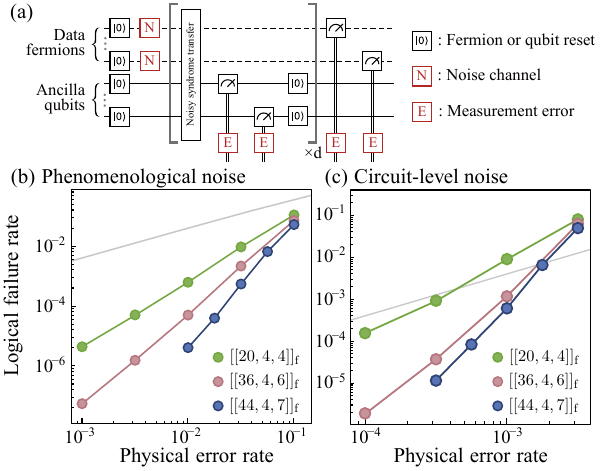}
\caption{
	(a) Illustration of the circuit-level noise model where 
	the syndrome extraction transfer circuits used to establish entanglement 
	between data fermionic sites and ancilla qubits are noisy.
	The initial states are physical $|0\rangle$ states for fermions.
	In the end, a transversal measurement on all 
	physical fermionic modes in the particle number basis is performed.
	Logical failure rate $p_L$ versus physical error rate $p$ under (b) phenomenological 
	and (c) circuit-level noise model for fermionic $[[20,4,4]]_{\text{f}}$, 
	$[[36,4,6]]_{\text{f}}$ and $[[44,4,7]]_{\text{f}}$ double-chain bicycle codes.
	We also perform a power law fit to the data according to 
	the fitting formula $p_L\sim p^\alpha$ with $\alpha=2.21, 3.06, 4.3$ in (b) 
	and $\alpha=2.12, 3.46, 4.54$ in (c), and
	the gray lines represent the probability that at least an error occurs on 
	four physical fermionic modes.
	Here, the Tesseract decoder~\cite{beni2025tesseractdecoder} is employed, and
	the results obtained using the belief propagation and ordered-statistical 
	decoding (BP+OSD)~\cite{liang2025low,roffe_decoding_2020}
	algorithm are presented in Supplemental Material.
}
\label{MemorySimulation}
\end{figure}

\begin{figure}[t]
\centering
\includegraphics[width=\linewidth]{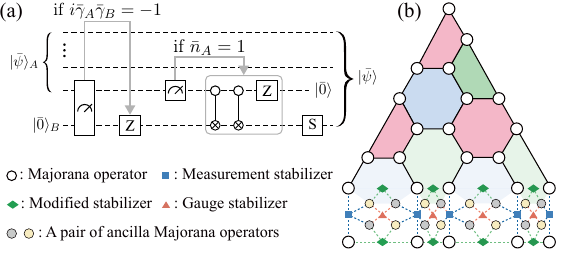}
\caption{
	(a) Quantum circuit to transfer a logical fermionic state $\ket{\overline{\psi}}_A$ 
	to code block $B$ (initialized as $\ket{\overline{0}}_B$). The protocol involves: 
	(i) measuring $\ii\overline{\gamma}_A\overline{\gamma}_B$ and applying the $Z$ 
	gate of $\exp(\ii\pi\overline{n}_B)$, provided the measurement result is $-1$; 
	(ii) measuring the logical particle number $\overline{n}_A$ and 
	applying two braiding gates 
	and one fermionic $Z$ gate,
	if the measurement outcome is $1$; 
	(iii) applying a phase rotation $S$ gate to the state in the processor.
	The validation of this circuit is provided in Supplemental Material.
	(b) Illustration of the first method for lattice surgery achieving fault-tolerant measurements of $\ii\overline{\gamma}_A \overline{\gamma}_B$ in the case with $|\overline{\gamma}_A|=|\overline{\gamma}_B|$.
	In this example, code $B$ is a $d=5$ fermionic color code, and code $A$ is a fermionic LDPC code containing five Majorana operators in the support of $\overline{\gamma}_A$.
	The empty circles represent the $\gamma$-type Majorana operators in blocks $A$ and $B$.
	Several pairs of ancilla Majorana operators are introduced with each horizontal pair represented by yellow and gray circles to denote $\gamma$-type and $\gamma^\prime$-type Majorana operator forming a complex fermion.
	The $\gamma$-type stabilizer generators at the boundary, connected to ancilla fermions, are modified to incorporate the ancilla Majorana operators (green rhombuses).
	Measurement and gauge stabilizer generators are depicted as blue squares and red triangles, respectively.
	The product of these stabilizer generators yields the joint logical operator $\ii\overline{\gamma}_A \overline{\gamma}_B$. The $\gamma^\prime$-type Majorana stabilizers remain unchanged.
}
\label{LatticeSurgery}
\end{figure}

Previous studies have primarily focused on encoding Majorana modes into qubits~\cite{Bravyi2010njp,litinski2018quantum,vijay2015majorana}.
While some results also demonstrate the existence of odd-weight logical operators in Majorana surface codes~\cite{mclauchlan2024new}, there is a lack of methods for systematically identifying these operators within a fermionic LDPC code.
To identify them, we generalize the Gram-Schmidt orthogonalization process to vector spaces over $\mathbb{F}_2$~\cite{prasetiaextended}.
Let $\left\{\vec{\gamma}_j|j=1,\dots,k^\prime, k^\prime \ge k \right\}$ be a basis for $\ker{A}/\text{im}(A^T)$, which is itself also a vector space; $\vec{\gamma}_j$ represents the corresponding logical $\overline{\gamma}_j$ or $\overline{\gamma}^\prime_j$ operator, expressed as $\gamma_1^{[\vec{\gamma}_j]_1}\dots \gamma_n^{[\vec{\gamma}_j]_n}$ or $(\gamma_1^\prime)^{[\vec{\gamma}_j]_1}\dots (\gamma_n^\prime)^{[\vec{\gamma}_j]_n}$.
If there exists a basis vector with odd number of $1'$s, say $\vec{\gamma}_1$, we remove it from the basis and make the remaining basis vectors orthogonal to it via $\vec{\gamma}_j\to \vec{\gamma}_j-(\vec{\gamma}_1\cdot \vec{\gamma}_j)\vec{\gamma}_1$, i.e., $\vec{\gamma}_1$ has even overlap with all others.
This procedure is repeated until no odd-weight basis vector remains.
The obtained odd-weight vectors correspond to logical operators which have odd weight, and the overlap weight between any two of them is even, thus satisfying the fermionic anticommutation relation.
The necessary and sufficient condition for the existence of at least one odd-weight logical is $(1,1,\dots,1)^T\notin \text{im}(A^T)$~\cite{Bravyi2010njp}.
We find that for all cases we consider, the procedure produces a linearly independent set containing $k^\prime$ odd-weight binary vectors (see Supplemental Material).

Based on four different LDPC codes, which include finite projective geometry codes~\cite{farinholt2012quantum},
finite Euclidean geometry codes~\cite{cao2012novel,aly2008class}, bicycle 
codes~\cite{MacKay2004ieee}, and stacked LDPC codes~\cite{liu2026stacked}, 
we construct four classes of fermionic LDPC codes, which 
yield the same code parameters 
$[[n,k,d]]_{\text{f}}$ as the original qubit codes (see Supplemental Material).
We now focus on three fermionic double-chain bicycle codes, which belong
to the fermionic stacked LDPC codes, and evaluate their logical 
information resilience as a fermionic memory through numerical simulations.
Specifically, we initialize all physical fermionic modes in the vacuum state, $|0\rangle$.
The performance of a code depends on quantum error models employed. For completeness,
we consider three types of error models, analogous to the qubit case: code-capacity,
phenomenological and circuit-level error models [see Fig.~\ref{MemorySimulation}(a)]~\cite{vasic2025quantum}.
The code-capacity error model does not account for errors in 
syndrome measurements and thus represents the upper bound 
of code performance against the errors considered.
In the phenomenological and circuit-level error models, 
the syndrome measurement circuits are noisy: In the former, only 
the measurement outcomes are noisy, whereas in the latter, all operations in 
the circuit are also noisy (see Supplemental Material).
Under these two error models, we perform $N_c=d$ rounds of 
syndrome measurements, followed by
a transversal measurement on all physical fermionic modes in the particle number basis, 
as illustrated in Fig.~\ref{MemorySimulation}(a).
We then apply a decoder to interpret the obtained syndromes 
and verify whether logical errors occur.

In numerical experiments, we perform $N_{\mathrm{sample}}$ trials and
identify $N_{\mathrm{error}}$ erroneous occurrences.  
Consequently, the logical error probability is given by 
$P_{\mathrm{L}}(N_c)=N_{\mathrm{error}}/N_{\mathrm{sample}}$.
Conventionally, the logical failure rate is defined as
$p_{\mathrm{L}}=1-(1-P_{\mathrm{L}}(N_c))^{1/N_c}\approx P_L/N_c$~\cite{xu2024constant,bravyi2024high},
which approximately characterizes the logical error rate per syndrome cycle.

Figure~\ref{MemorySimulation}(b) and (c) show the logical failure rate 
with respect to the physical error rate $p$ for three fermionic double-chain bicycle codes.
We clearly see that the logical failure rate is significantly suppressed compared to 
the physical one according to $p_{\mathrm{L}} \sim p^\alpha$ with $\alpha \approx \lceil \frac{d}{2} \rceil$.
By convention~\cite{bravyi2024high}, we define the code's pseudo-threshold 
as the solution to the break-even equation $p_{\mathrm{L}}(p) = P(p, k)$, 
where $P(p, k)=1-(1-p)^k$ is the probability of at least one error occurring on $k$ physical 
fermionic modes at physical error rate $p$.
Under the circuit-level error model, the calculated pseudo-thresholds of the $[[36,4,6]]_{\text{f}}$ and 
$[[44,4,7]]_{\text{f}}$ codes exceed $0.1\%$.

In Supplemental Material,
we also provide the numerical simulation results for fermionic finite projective geometry codes,
fermionic finite Euclidean geometry codes, fermionic bicycle codes, fermionic
double-layer bivariate bicycle (BB) codes, fermionic double-layer twisted BB codes,
and fermionic double-layer reflection codes.
In addition, we study several other self-dual codes such as Kitaev Majorana 
code~\cite{kitaev2006anyons} and unicycle codes~\cite{MacKay2004ieee}
and find that they are not suitable for producing well-behaved fermionic LDPC 
because they do not support the generation of odd-weight logical operators.

We now study how to perform logical operations on the fermionic LDPC memory.
Instead of executing logical gates directly on the fermionic LDPC, we consider a method whereby the states of a target logical fermionic mode in memory are transferred to an external fermionic color code processor.
Logical operations are then performed on these processors, and the states are subsequently transferred back to the memory.
This approach is inspired by fault-tolerant measurements of logical operators in qubit LDPC codes~\cite{horsman2012surface,fowler2018low,xu2024constant,cohen2022low,poulsen2017fault}.
If the logical operation involves two logical fermionic modes, we consider two color code processors.
The logical information of the two fermionic modes is individually transferred, logical operations are applied between these two color code blocks, and the information is subsequently returned to memory.
Logical operations involving more logical fermionic modes can be performed similarly.
For fermionic state transfers between memory and processor, we design a measurement-based circuit as illustrated in Fig.~\ref{LatticeSurgery}(a).
This circuit enables the transfer of a single logical fermionic-site information described by logical Majorana operators $\overline{\gamma}_A$ and $\overline{\gamma}_A^\prime$ to a processor described by logical Majorana operators $\overline{\gamma}_B$ and $\overline{\gamma}_B^\prime$, which is initialized in the logical vacuum state $\ket{\overline{0}}_B$.
In other words, an output state is the same as the initial state except that the processor logical fermionic mode plays the role of the corresponding logical fermionic mode in memory.
During the transfer, a key step is the fault-tolerant measurement of the joint logical operator $\ii\overline{\gamma}_A \overline{\gamma}_B$, as shown in Fig.~\ref{LatticeSurgery}(a), which is based on general fermionic lattice surgery we describe below while the other operations can be carried out transversally.

\begin{figure}[t]
\centering
\includegraphics[width=\linewidth]{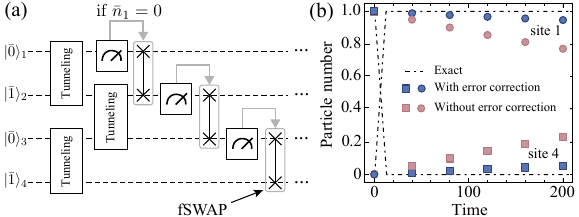}
\caption{
	(a) Quantum circuit used to benchmark our protocol.
	We first initialize the system in the logical state $\left|\overline{0}\overline{1}\overline{0}\overline{1}\right\rangle$ 
	and then apply three tunneling gates with each gate executing the operation ${T}_{j,j+1}= \exp(\ii \pi (\overline{c}_j^\dagger \overline{c}_{j+1} + \text{H.c.})/2)$.
	Onsite measurements of $\overline{n}_j=\overline{c}_j^\dagger \overline{c}_j$ with $j=1,2,3$ are then performed, followed by a fermionic swap (fSWAP) gate if the measurement result is zero.
	The unitary gates and feedback process constitute a module, which is executed consecutively multiple times.
	(b) The expectation value of logical particle numbers at site $1$ and $4$ 
	with respect to time which is characterized by the number of executed modules.
	We see that the simulation under circuit-level noise with 
	error correction (filled blue circles and squares) is close to the noiseless simulation result 
	(the dotted-dashed lines), in stark contrast to the simulation results under noise 
	without error correction (filled red circles and squares). 
	In the simulation, we apply the first method for lattice surgery. Results using the 
	second method are provided in the Supplemental Material.
	Here, we set the physical error rate $p=10^{-4}$.
	The error bars are hidden behind the symbols.
}
\label{TaskSimulation}
\end{figure}

We propose two methods for fermionic lattice surgery as shown in Fig.~\ref{LatticeSurgery}(b), Fig.~\ref{figs-stabilizers}, and Fig.~\ref{figs-general-LS}.
The first method applies when code $B$ is a fermionic color code 
with a distance $d_B$ such that $d_B=d_A$, where 
$d_A=|\overline{\gamma}_A|=|\overline{\gamma}_A^\prime|$ with $|\dots|$ 
denoting the weight of a Majorana operator string.
The second method is applicable when either $d_B=d_A$ or 
$d_B \neq d_A$, and when code $B$ is either a fermionic color code
or a fermionic LDPC code.
The procedure proceeds by introducing ancilla Majorana operators between 
the two codes, modifying the stabilizers associated with $\overline{\gamma}_A$ 
and $\overline{\gamma}_B$, and adding new stabilizers that merge $A$ 
and $B$ into a larger fermionic LDPC code $\mathcal{C}_{\text{merged}}$. 
Since the joint operator $\ii\overline{\gamma}_A\overline{\gamma}_B$ 
belongs to the stabilizer group $\mathcal{S}_{\text{merged}}$ of $\mathcal{C}_{\text{merged}}$,
it can be measured 
fault-tolerantly by collecting the outcomes of the stabilizer measurements 
in $\mathcal{C}_{\text{merged}}$.
Figure~\ref{LatticeSurgery} provides an example of fermionic lattice surgery
using method 1. 
In this example, we align the support of the logical 
operator $\overline{\gamma}_A$ and $\overline{\gamma}_B$ and 
introduce a set $Q_C$ of $4(d_A-1)$ ancilla Majorana operators 
(corresponding to $2(d_A-1)$ fermionic modes) between these two codes, 
labeled as $\gamma_{a,1},\gamma_{a,1}^\prime,\dots,\gamma_{a,(d_A-1)},\gamma_{a,(d_A-1)}^\prime$ and $\gamma_{b,1},\gamma_{b,1}^\prime,\dots,\gamma_{b,(d_A-1)},\gamma_{b,(d_A-1)}^\prime$.
The original $\gamma$-type stabilizer generators at the boundary are modified 
by including two ancilla Majorana operators.
Measurement stabilizer generators $M_1,M_2,\dots,M_{d_A}$ and gauge 
stabilizer generators $G_1,G_2,\dots,G_{d_A-1}$ are also introduced.
The original $\gamma^\prime$-type stabilizer generators for both blocks 
of $A$ and $B$ are not modified.
This merging procedure creates a merged stabilizer code $\mathcal{C}_{\mathrm{merged}} $ 
where $\ii\overline{\gamma}_A \overline{\gamma}_B$ becomes a stabilizer.
In the End Matter, we provide guidelines for introducing ancilla Majorana 
operators and constructing stabilizer generators for both methods.
In Supplemental Material, we show that the distance of dressed logical operators is not decreased based on the formalism of the subsystem code~\cite{vuillot2019code,aly2006subsystem}, and numerical simulation demonstrates fault-tolerance during merge and split under circuit-level error model.

We now demonstrate the capability of our codes by simulating the 
quantum circuit of four fermionic modes shown in Fig. \ref{TaskSimulation}(a).
The circuit incorporates tunneling operations that enable the tunneling of fermions between neighboring sites, and onsite measurements, followed by fermionic swap gates, provided the measurement outcome is zero.
For an initial state $\ket{\overline{0}\overline{1}\overline{0}\overline{1}}$, we expect that the steady state of this dynamical process is a skin state
where fermionic particles mainly reside in the upper half due to feedback effects~\cite{wang2024absence,feng2023absence,liu2024dynamical}.
To simulate the dynamical behavior, we employ the fermionic $[[88,4,7]]_{\text{f}}$ double-chain bicycle code 
as our fermionic LDPC code and use the 
four logical fermionic modes without overlap for the simulation.
Additionally, we introduce two copies of the fermionic $[[37,1,7]]_{\text{f}}$ color code to function as processors.
The calculated time evolution of the particle number at the first and fourth site is shown in Fig.~\ref{TaskSimulation}(b).
The value of $\left\langle \overline{n}_1\right\rangle$ suddenly increases to one while the value of $\left\langle \overline{n}_4\right\rangle$ decreases to zero, which aligns with the expected characteristics of a skin state.
Notably, with our fault-tolerant logical operations and error correction techniques, the error-corrected data demonstrate significantly higher fidelity compared to the uncorrected data, as shown in Fig. \ref{TaskSimulation}(b).

In summary, we have developed a systematic framework for fault-tolerant 
fermionic quantum computation based on fermion-to-fermion LDPC codes. 
This realization relies on the proposed methods for executing fermionic lattice surgery, 
where a general lattice-surgery protocol can be applied between any two fermionic LDPC codes.
Such fermion-to-fermion quantum computation has the potential to significantly reduce the computational complexity associated with simulating
fermionic systems.
Our work encourages the pursuit of high-performance fermionic LDPC codes with higher encoding rates and larger code distances, as well as the advancement of more efficient methods for executing logical operations with reduced overhead.

\textit{Note added.}
During the final stage of this work, we became aware of a related work~\cite{Mudassar2025arxiv}, which uses Majorana LDPC codes to encode logical qubits.

\emph{Acknowledgments}---We thank Y. Li, Q. Xu, F. Wei, and Y. Mao for helpful discussions.
This work is supported by the National Natural Science Foundation of China (Grant No. 12474265 and No. 11974201) and Innovation Program for Quantum Science and Technology (Grant No. 2021ZD0301604).
We also acknowledge the support by center of high performance computing, Tsinghua University.

\emph{Data availability}---The data that support the findings of this article are openly available (ref.~\cite{dataref}).

\onecolumngrid
\begin{center}
{\textbf{End Matter}}
\end{center}

\twocolumngrid
{
In the End Matter, we will provide the guidelines for introducing ancilla Majorana 
operators and constructing stabilizer generators for fermionic lattice surgery.

For clarity, we use the following terminology: 
additional Majorana operators are termed ancilla Majorana operators,
stabilizers inherited from the original codes and modified 
during the merging are called modified stabilizers, 
newly introduced stabilizers acting solely on ancilla Majorana operators are termed 
gauge stabilizers, and stabilizers involving Majorana operators from both codes $A$ 
and $B$ are referred to as measurement stabilizers.

\emph{Method 1}.---Since logical operators in our fermionic LDPC code all have odd weight, 
for any $\overline{\gamma}_A$ in a fermionic LDPC code memory $A$, 
we can always select a triangular color code $B$~\cite{Bravyi2010njp,Schuckert2024arxiv} encoding a single pair of 
logicals $(\overline{\gamma}_B,\overline{\gamma}_B^\prime)$, 
with $|\overline{\gamma}_A|=|\overline{\gamma}_B|=|\overline{\gamma}_B^\prime|=d$.
Method 1 as detailed in the following is applicable to the lattice surgery between a fermionic LDPC code 
and a fermionic color code with $|\overline{\gamma}_A|=|\overline{\gamma}_B|$.

We align the support of $\overline{\gamma}_A$ and $\overline{\gamma}_B$ 
in a pair of parallel lines, each consisting of $d$ points, 
as shown in Fig.~\ref{LatticeSurgery}(b), and indexing them as 
$\left\{\gamma_{A,i}\right\}_{i=1}^{d}$ and 
$\left\{\gamma_{B,i}\right\}_{i=1}^{d}$ from left to right, respectively. 
We also index the checks supported on $A$ $(B)$ that are associated with 
$\overline{\gamma}_A$ ($\overline{\gamma}_B$) as $\left\{c_{A,i}\right\}_{i=1}^{N_{A}}$ 
($\left\{c_{B,i}\right\}_{i=1}^{N_{B}}$),
where $N_A$ and $N_B=d-1$ are the number of stabilizers associated with 
$\overline{\gamma}_A$ and $\overline{\gamma}_B$, respectively. 
As shown in Fig.~\ref{figs-stabilizers}, each $ c_B$ forms a plaquette on the boundary of the 
color code lattice.
We now introduce ancilla Majorana modes, modify the original stabilizers, 
and introduce new stabilizers step by step.

\begin{figure*}
	\centering 
	\includegraphics[width=0.8\textwidth]{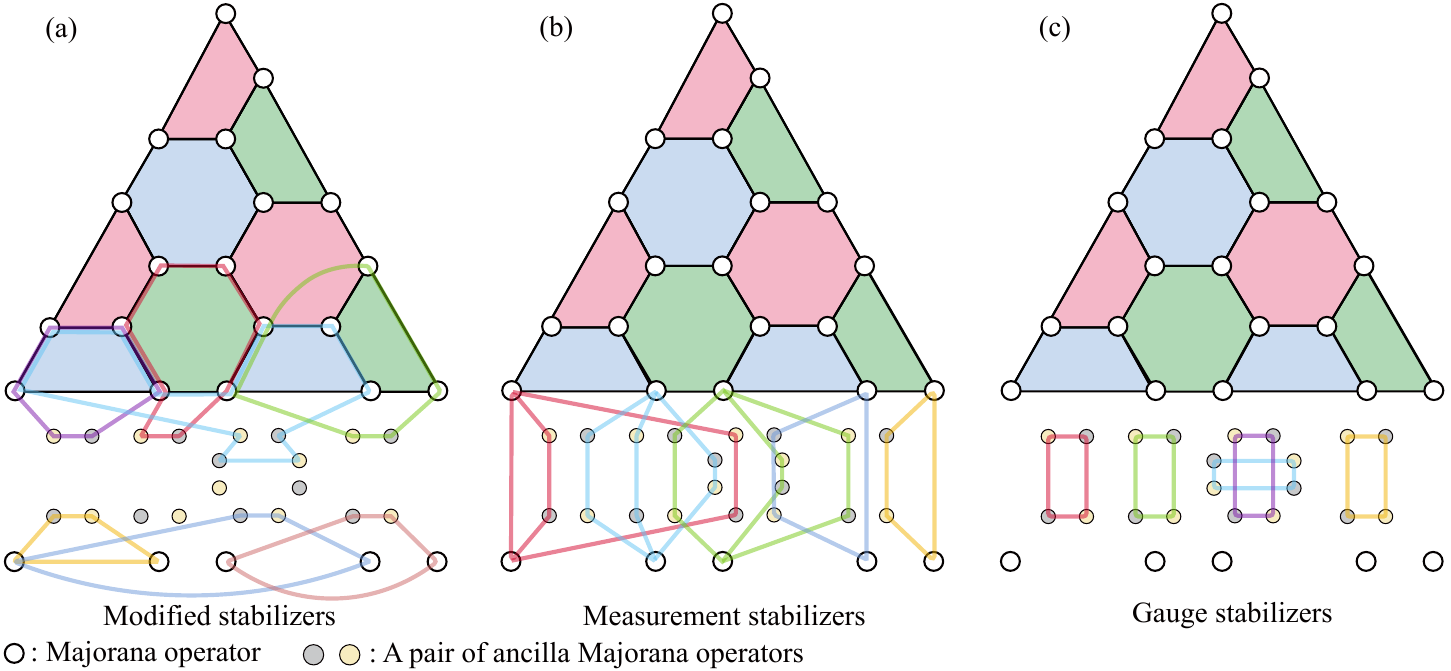}
	\caption{Illustration of modified stabilizers in (a), measurement stabilizers in (b),
		and gauges stabilizers in (c) for a specific case of fermionic lattice surgery using method 1.
		Before merging, the original stabilizers include those of the fermionic LDPC code, 
		the color code, and the stabilizers of ancilla Majorana operators (each stabilizer 
		includes a pair of ancilla Majorana operators forming a complex fermion).} 
	\label{figs-stabilizers}
\end{figure*}

We first modify all the stabilizers $c_{A,i}$ and $c_{B,i}$ by introducing ancilla Majorana operators. 
Specifically, consider each 
\begin{equation}
	(c_{A,i})_{\overline{A}}\sim \gamma_{A,j_1}\gamma_{A,j_2}\dots\gamma_{A,j_{m_{i}}},
\end{equation}
where $1\leq j_1< \dots <  j_{m_{i}}\leq d$, $\sim$ denotes equality up to a phase,
${\overline{A}}$ denotes the support of $\overline{\gamma}_A$,
$(\cdots)_{\overline{A}}$ is the operator constrained on ${\overline{A}}$,
and $m_{i}=|(c_{A,i})_{\overline{A}}|$ is an even number, ensured by
the commutativity between $\overline{\gamma}_A$ and $c_{A,i}$.
We introduce $m_{i}$ ancilla Majorana operators
$\left\{\gamma_{a,i,j_1},\gamma_{a,i,j_2}^\prime,\dots,\gamma_{a,i,j_{m_{i}-1}},\gamma_{a,i,j_{m_{i}}}^\prime\right\}$,
where the subscript $a$ indicates that these modes are associated with the stabilizer of $A$,
$i$ is the index of the stabilizer $c_{A,i}$,
and $j_1,j_2,\dots,j_{m_{i}}$ correspond to the positions of the Majorana operators in $c_{A,i}$
in the support of $\overline{\gamma}_A$.
The last two subscripts are used to index the order of the ancilla modes,
corresponding to the second subscript in $\gamma_{a,j}$ and 
$\gamma_{a,j}^{\prime}$ with $1\leq j\leq d_A-1$ as described in the main text.
We modify the stabilizer $c_{A,i}$ to 
\begin{equation}
	c_{A,i}^\prime\sim c_{A,i} \gamma_{a,i,j_1}\gamma_{a,i,j_2}^\prime\dots\gamma_{a,i,j_{m_{i}-1}}\gamma_{a,i,j_{m_{i}}}^\prime.
\end{equation}
Similarly, we introduce $m_{i}$ ancilla Majorana operators
$\left\{\gamma_{b,i,j_1},\gamma_{b,i,j_2}^\prime,\dots,\gamma_{b,i,j_{m_{i}-1}},\gamma_{b,i,j_{m_{i}}}^\prime\right\}$ 
for code $B$. 
For each pair $( \gamma_{a,i, j_{2n-1}}, \gamma_{a,i, j_{2n}}^\prime)$ with $1\leq n\leq m_i/2$: 
\begin{itemize}
	\item If $j_{2n-1}+1=j_{2n}$,
	we introduce a modified stabilizer in code $B$ as
	$ c_{B,i}^\prime\sim c_{B,i} \gamma_{b,i,j_{2n-1}} \gamma_{b,i,j_{2n}}^\prime$, 
	and a gauge stabilizer as $\gamma_{a,i,j_{2n-1}} \gamma_{a,i,j_{2n}}^\prime\gamma_{b,i,j_{2n-1}} \gamma_{b,i,j_{2n}}^\prime$.
	\item If $j_{2n-1}+2=j_{2n}$, 
	this modified stabilizer is defined as
	$ c_{B,i}^\prime\sim c_{B,j_{2n-1}} c_{B,j_{2n}-1} \gamma_{b,i,j_{2n-1}} \gamma_{b,i,j_{2n}}^\prime$,
	and the gauge stabilizer is $\gamma_{a,i,j_{2n-1}} \gamma_{a,i,j_{2n}}^\prime\gamma_{b,i,j_{2n-1}} \gamma_{b,i,j_{2n}}^\prime$.
	\item Otherwise,
	we add 4 additional Majorana modes 
	$\left\{ \gamma_{b,i,j_{2n-1}+1}^\prime,\ \gamma_{b,i,j_{2n}-1},\ \gamma_{a,i,j_{2n-1}+1},\ \gamma_{a,i,j_{2n}-1}^\prime\right\}$, 
	let the modified stabilizer be
	$ c_{B,i}^\prime\sim c_{B,j_{2n-1}} c_{B,j_{2n}-1} \gamma_{b,i,j_{2n-1}}\gamma_{b,i,j_{2n-1}+1}^\prime \gamma_{b,i,j_{2n}-1} \gamma_{b,i,j_{2n}}^\prime $,
	and introduce two gauge stabilizers as $\gamma_{a,i,j_{2n-1}} \gamma_{a,i,j_{2n}}^\prime\gamma_{b,i,j_{2n-1}} \gamma_{b,i,j_{2n}}^\prime$, 
	$\gamma_{a,i,j_{2n-1}+1} \gamma_{a,i,j_{2n}-1}^\prime\gamma_{b,i,j_{2n-1}+1}^\prime \gamma_{b,i,j_{2n}-1}$.
\end{itemize}

We repeat this process for each $c_{A,i}$. 
We note that for those unused $ c_{B,i}$ after iterating over all $c_{A,i}$, 
say the number of which as $M_B$.
For each such $c_{B,m}$, where $1\leq m\leq M_B$ represents the index of $c_{B,m}$,  
we further introduce four Majorana operators 
$\left\{\gamma_{b,N_A+m,i},\ \gamma_{b,N_A+m,i+1}^\prime,\ \gamma_{a,N_A+m,i},\ \gamma_{a,N_A+m,i+1}^\prime\right\}$,
define a modified stabilizer $c_{B,i}^\prime$ as
$c_{B,i} \gamma_{b,N_A+m,i} \gamma_{b,N_A+m,i+1}^\prime$, 
and add a gauge stabilizer as $\gamma_{b,N_A+m,i} \gamma_{b,N_A+m,i+1}^\prime \gamma_{a,N_A+m,i}\gamma_{a,N_A+m,i+1}^\prime$.

After that, we define $d$ measurement stabilizers $\left\{ M_i\right\}_{i=1}^d$ as
\begin{equation}
	M_i\sim \gamma_{A,i} \gamma_{B,i}\prod_n  \gamma_{a,n,i}^{(\prime)} \gamma_{b,n,i}^{(\prime)},
\end{equation}
where $\prod_n  \gamma_{a,n,i}^{(\prime)} \gamma_{b,n,i}^{(\prime)}$  
denotes the product of all ancilla Majorana operator (either of $\gamma$-type or $\gamma^\prime$-type) 
at site $i$
with $n$ running over the subscripts of the ancilla modes associated with $\gamma_{A,i}$ and $\gamma_{B,i}$.
In other words, the measurement stabilizer at site $i$ consists of
$\gamma_{A,i}$, $\gamma_{B,i}$ and all ancilla Majorana operators at this position.  
We note that $\ii\overline{\gamma}_A\overline{\gamma}_B$ is proportional to 
the product of all measurement stabilizers and the gauge stabilizers,
since all $\gamma_{A,i}$ and $\gamma_{B,i}$ appear once
and all ancilla Majorana operators appear twice in this product and thus cancel out.
Figure~\ref{figs-stabilizers} displays all modified, measurement, and
gauge stabilizers for an illustrative case, which is more complicated than 
that in Fig.~\ref{LatticeSurgery}(b).

\emph{Method 2}.---We now introduce the second method for fermionic lattice surgery 
applicable when $|\overline{\gamma}_A|=|\overline{\gamma}_B|$ or 
$|\overline{\gamma}_A|\neq|\overline{\gamma}_B|$, and when code $B$ is either a fermionic color code
or a fermionic LDPC code (see Supplemental Material for Pseudo-code).
As an example depicted in Fig.~\ref{figs-general-LS}(a), we align a weight-$5$ logical operator $\overline{\gamma}_A$ 
of fermionic LDPC code $A$ and a weight-$3$ logical operator $\overline{\gamma}_B$ of fermionic color code $B$
in two parallel lines, and index the Majorana modes in these two lines from top to bottom as
$\left\{\gamma_{A,i}\right\}_{i=1}^{5}$ and $\left\{\gamma_{B,i}\right\}_{i=1}^{3}$, respectively.

\begin{figure*}
	\centering 
	\includegraphics[width=\textwidth]{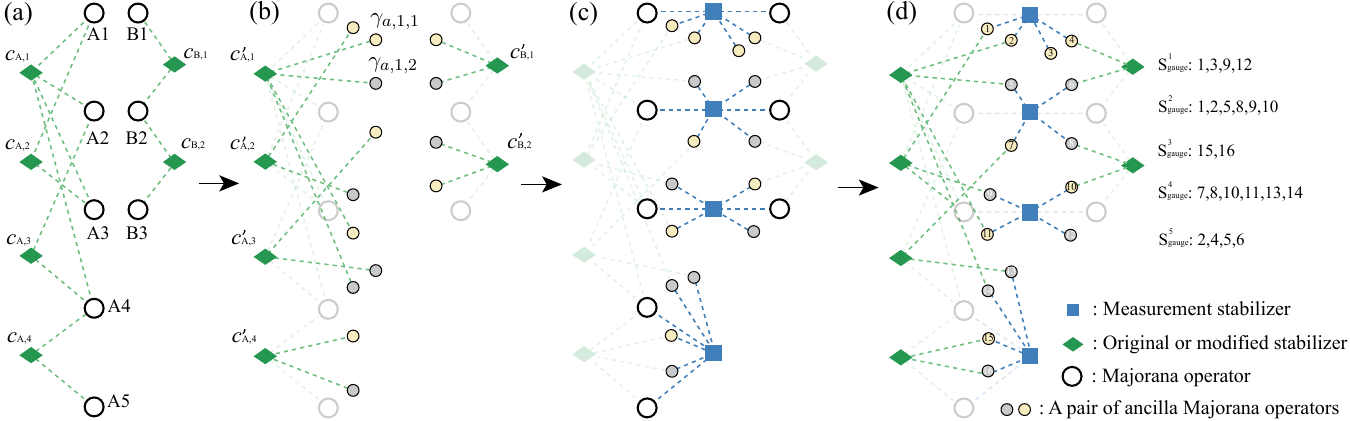}
	\caption{Illustration of how to add ancilla Majorana modes 
		and stabilizers for a merged code using method 2. Here we consider 
		two logical operators 
		in codes A and B, supported by five and three Majorana operators, respectively.
		Initially, we have two logical operators described by the original 
		stabilizer generators represented by green rhombuses and 
		Majorana operators represented by circles
		for 
		code $A$ in the support of $\overline{\gamma}_A$ 
		and for code $B$ in the support of $\overline{\gamma}_B$ (see (a)). 
		Pairs of ancilla Majorana operators (yellow and gray circles) are added following the rules 
		detailed in the text, and the modified stabilizer generators are introduced (see (b)).
		We then incorporate measurement stabilizer generators with an additional pair of
		ancilla Majorana operators to make their weights even (see (c)). 
		In (d), we label all ancilla modes and construct
		a graph with vertices corresponding to stabilizers and ancilla Majorana modes, 
		and edges defined between them accordingly, which 
		corresponds to Fig.~S3(a) in Supplemental Material. The gauge stabilizer generators are also presented 
		in Fig.~S3(c) in Supplemental Material.
	} 
	\label{figs-general-LS}
\end{figure*}

We first modify the original stabilizers associated with $\overline{\gamma}_A$ and $\overline{\gamma}_B$.
As shown in Fig.~\ref{figs-general-LS}(a), there are $N_A=4$ stabilizers for $\overline{\gamma}_A$ and $N_B=2$ stabilizers 
for $\overline{\gamma}_B$, whose weights of overlap with $\overline{\gamma}_A$ and $\overline{\gamma}_B$
are denoted as 
$\left\{m_{A,i}\right\}_{i=1}^{N_A}$ and $\left\{m_{B,i}\right\}_{i=1}^{N_B}$, respectively.
For each stabilizer $c_{A,i}$ and $c_{B,i}$, we introduce $m_{A,i}$ and $m_{B,i}$ ancilla Majorana operators
$\{\gamma_{a,i,j_1},\gamma_{a,i,j_2}^\prime,\cdots,\gamma_{a,i,j_{m_{A,i}-1}},\gamma_{a,i,j_{m_{A,i}}}^\prime\}$
and $\{\gamma_{b,i,j_1},\gamma_{b,i,j_2}^\prime,\cdots,\gamma_{b,i,j_{m_{B,i}-1}},\gamma_{b,i,j_{m_{B,i}}}^\prime\}$,
respectively,
where $j_k$ correspond to the vertical positions of the Majorana operators in $c_{A,i}$ or $c_{B,i}$.
We modify $c_{A(B),i}$ to $c_{A(B),i}^\prime$ by multiplying these ancilla Majorana operators as before,
as shown in Fig.~\ref{figs-general-LS}(b).
The difference is that, in this case, the modified stabilizer $c_{B,i}^\prime$ is constructed from a single
stabilizer $c_{B,i}$, instead of the product of several stabilizers in the previous case.

Next, we define the measurement stabilizers $\{M_i\}_{i=1}^{D_{AB}}$, where 
$D_{AB}=({|\overline{\gamma}_A|+|\overline{\gamma}_B|})/{2}$. 
For the first $\min\{|\overline{\gamma}_A|, |\overline{\gamma}_B|\}$ measurement stabilizer $M_i$, 
we define $M_i$ as the product of $\gamma_{A,i}$, $\gamma_{B,i}$,
and all ancilla Majorana modes at site $i$, that is,
\begin{equation}
	M_i\sim \gamma_{A,i} \gamma_{B,i} \prod_{n_a,n_b}\gamma_{a,n_a,i}^{(\prime)} \gamma_{b,n_b,i}^{(\prime)},
\end{equation}
where $\prod_{n_a,n_b}\gamma_{a,n_a,i}^{(\prime)} \gamma_{b,n_b,i}^{(\prime)}$ denotes the product of
all ancilla Majorana operators at site $i$ (either of $\gamma$-type or $\gamma^\prime$-type)
with $n_a$ and $n_b$ running over the subscripts of the ancilla modes associated with $\gamma_{A,i}$ 
and $\gamma_{B,i}$, respectively.
For the remaining $|(|\overline{\gamma}_A|-|\overline{\gamma}_B|) |/{2}$ measurement stabilizers 
$\{M_i\}_{i=\min\{|\overline{\gamma}_A|,|\overline{\gamma}_B|\}+1}^{D_{AB}}$,
they are defined as the product of a pair of Majorana operators in $\{\gamma_{A,i}\}$ or $\{\gamma_{B,i}\}$,
and the ancilla Majorana modes associated with them.
For example, in Fig.~\ref{figs-general-LS}(c), we define 
$M_4\sim \gamma_{A,4}\gamma_{A,5}\gamma_{a,1,4}^\prime\gamma_{a,3,4}^\prime\gamma_{a,4,4}\gamma_{a,4,5}^\prime$.
We note that if any $M_i$ has odd weight, we introduce an additional ancilla mode
$\gamma_{p,i}$, and modify $M_i$ to $M_i\gamma_{p,i}$ up to a phase factor to ensure even weight, as shown in the first 
and third measurement stabilizers in Fig.~\ref{figs-general-LS}(c).
The number of such additional ancilla modes is denoted as $N_P$, which must be even since the product of
all measurement stabilizers is an even-weight operator. 
Gauge stabilizers are then identified from a graph theory perspective,
as detailed in Supplemental Material.

}

\begin{widetext}
	\setcounter{equation}{0} \setcounter{figure}{0} \setcounter{table}{0} 
	\renewcommand{\theequation}{S\arabic{equation}} \renewcommand{\thefigure}{S\arabic{figure}}
In the Supplemental Material, we review Majorana stabilizer codes and 
detail our method for identifying odd-weight logical operators in a fermionic LDPC 
code in Section S-1, validate the quantum state teleportation circuit in Section S-2, 
analyze the overhead of fermionic lattice surgery and prove the fault-tolerance of 
this procedure in Section S-3, and finally provide more details regarding calculations 
of the logical failure rate of the fermionic LDPC memory and simulations of fermionic circuits in Section S-4.

\section{S-1. Construction of fermionic LDPC codes}\label{sec:CodeConstruction}

In this section, we will follow Ref.~\cite{Bravyi2010njp} to review Majorana stabilizer codes and detail our method for identifying odd-weight logical operators in a fermionic LDPC code.

\subsection{A. Majorana stabilizer codes}

In this subsection, we will review the Majorana stabilizer codes (also see Ref.~\cite{Bravyi2010njp}).
Consider the group of Majorana operators, 
$\text{Maj}(2n)\equiv\left\{{\Gamma}=\eta\prod_{j=1}^{n} {\gamma}_j^{\alpha_j} ({\gamma}_j^\prime)^{\alpha_j^\prime}\mid
\alpha_j,\alpha_j^\prime \in \{0,1\}, \eta \in \{\pm 1, \pm \ii\} \right\}$,
where $\gamma_j$ and $\gamma_j^\prime$ with $j=1,\dots,n$ are Majorana operators  
satisfying 
$\{{\gamma}_i, {\gamma}_j\} = 2\delta_{ij}$, 
$\{{\gamma}_i^\prime, {\gamma}_j^\prime\} = 2\delta_{ij}$,
$\{{\gamma}_i, {\gamma}_j^\prime\} = 0$.
These Majorana operators arise from complex fermion's creation (denoted $c_j^\dagger$) and annihilation
operators (denoted $c_j$) via
$\gamma_j=c_j+c_j^\dagger$ and $\gamma_j^\prime=\ii (c_j-c_j^\dagger)$.
Analogous to Pauli stabilizer codes being defined based on the Pauli group~\cite{nielsen2010quantum}, 
Majorana stabilizer codes are defined based on the group of Majorana operators~\cite{Bravyi2010njp}. 

To handle the group multiplication structure more systematically, 
we consider a $2n$-dimensional vector space $\mathbb{F}_2^{2n}$, 
where $\mathbb{F}_2$ is the binary field $\{0, 1\}$.
We then construct an isomorphism $\phi: \text{Maj}(2n) \to \mathbb{F}_2^{2n} \otimes \{\pm 1, \pm \ii\}$, 
mapping each element ${\Gamma}$ to a binary vector $\vec{\gamma} \in \mathbb{F}_2^{2n}$ and a phase $\eta$,
\begin{equation}
\phi(\Gamma) = (\vec{\gamma}, \eta) \quad \text{where} \quad \vec{\gamma} = (\alpha_1, \alpha_1^\prime, \alpha_2, \alpha_2^\prime, \dots, \alpha_{n}, \alpha_{n}^\prime)^T.
\end{equation}
To ensure that $\text{Maj}(2n) \cong \mathbb{F}_2^{2n} \times \{\pm 1, \pm \ii\}$, we define
$(\vec{\gamma}_1, \eta_1) \cdot (\vec{\gamma}_2, \eta_2) \equiv (\vec{\gamma}_1 \oplus \vec{\gamma}_2, \eta_1 \eta_2 f(\vec{\gamma}_1, \vec{\gamma}_2)),$ where $\vec{\gamma}_1 \oplus \vec{\gamma}_2$ is the standard component-wise addition modulo $2$ on $\mathbb{F}_2^{2n}$, equivalent to the XOR operation, and the factor $f(\vec{\gamma}_1, \vec{\gamma}_2)\in \{1,-1\}$, termed the sign function, encodes a sign arising from the inherent anticommutation relations of the Majorana operators.
Based on a common operator ordering convention, we define the sign function as $f \equiv (-1)^{\sum_{k} [\vec{\gamma}_2]_k m_k}$, where $m_k$ counts the number of nonzero elements in $\{ [\vec{\gamma}_1]_k,[\vec{\gamma}_1]_{k+1},\dots, [\vec{\gamma}_1]_{2n}\}$.

For any elements $\Gamma_1,\Gamma_2 \in \text{Maj}(2n)$, they either commute or anticommute.
Specifically, we have $\Gamma_1 \Gamma_2 =\Gamma_2 \Gamma_1  (-1)^{\text{p}(\vec{\gamma}_1)\text{p}(\vec{\gamma}_2) \oplus \vec{\gamma}_1 \cdot \vec{\gamma}_2}$~\cite{Bravyi2010njp}, where $\vec{\gamma}_j$ represents the vector mapped from $\Gamma_j$, and $\text{p}(\vec{\gamma})$ is a parity function denoting the parity of the number of nonzero entries in the vector $\vec{\gamma}$: 0 indicates an even parity and 1 indicates an odd parity.
We see that two even-weight Majorana strings commute if their 
overlap has even weight and anticommute if their overlap has odd weight. 
One thus can encode logical Majorana operators as even-weight physical Majorana operator strings.
However, if two logical Majorana operators are located in different blocks, which have no
overlap, then odd-weight strings must be considered, as they anticommute with each other.

The Majorana stabilizer code is defined by a Majorana stabilizer group 
$\mathcal{S}_{\text{maj}} = \langle {g}_1, \dots, {g}_m \rangle$ generated by independent 
Majorana operator strings  
$g_j\in \text{Maj}(2n)$ with $j=1,\dots,m$. Each element in the group commutes with any other elements in it.
The code space $\mathbb{H}_L$ is a subspace spanned by all state vectors that are simultaneous 
eigenstates of all stabilizer operators corresponding to eigenvalue $+1$~\cite{nielsen2010quantum}. The check matrix is 
defined as
\begin{equation}
H =\left(\begin{array}{c}
\vec{g}_{1}^T
\\ \vdots
\\ \vec{g}_{m}^T
\end{array} \right),
\end{equation}
where the vector $\vec{g}_{j}$ with $j=1,\dots,m$ corresponds to a stabilizer generator $g_j$. 
To preserve the parity of a physical fermion system, we suppose that all stabilizers have 
even weight. In this case, $H$ is a parity-check matrix of a self-orthogonal linear code, 
satisfying $HH^T = 0$. 
Logical operators correspond to operators that commute with all stabilizers but are
not in the stabilizer group. Mathematically, a set of all logical operators 
$\mathcal{L}=\mathcal{C}(\mathcal{S}) \setminus \mathcal{S}$,
where $\mathcal{C}(\mathcal{S})$ denotes the centralizer whose elements commute with 
all elements in $\mathcal{S}$. The code distance is determined by the minimum
weight in $\mathcal{L}$.

We can use weakly-self-dual CSS codes on qubits to construct Majorana stabilizer codes~\cite{Bravyi2010njp,Schuckert2024arxiv}.
Let $A$ be the parity-check matrix of a weakly-self-dual CSS code with parameters 
$[[n, k, d]]$, satisfying $A A^T = 0$~\cite{MacKay2004ieee}. 
We write $A=(\vec{u}_1,\vec{u}_2,\dots,\vec{u}_{m/2})^T$, 
where $m$ is an even positive integer, and
suppose that the row vectors in $A$ are linearly independent, corresponding to 
stabilizer generators.
Consequently, $\{\vec{u}_1,\vec{u}_2,\dots, \vec{u}_{m/2}, \vec{\gamma}_1,\dots,\vec{\gamma}_{k} \}$
forms a basis for $\ker{(A)}$ with $\left\{\vec{\gamma}_j\mid j=1,\dots,k \right\}$ 
corresponding to independent logical operators.
The set of logical operators is expressed as $L=\{\sum_{j=1}^{m/2} a_j \vec{u}_j + \sum_{j=1}^{k} b_j \vec{\gamma}_j \mid a_j,b_j\in\{0,1\} 
\text{ and } (b_1,b_2,\dots,b_{k})^T \neq \vec{0} \}$ with each vector 
contributing an $X$ and a $Z$ logical operator.
The code distance is thus determined by the minimum weight of the vectors in $L$.
We now use the direct sum $H = A \oplus A$ to form a Majorana stabilizer code's check matrix
for $n$ complex fermion modes. 
The code can be represented by a Tanner graph which has mirror 
reflection symmetry, where the blue and green squares 
correspond to $\gamma$-type stabilizers and $\gamma'$-type stabilizers, respectively,
as shown in Fig.~\ref{fig:graph-representation}(a).
The set $L$ also describes all logical operators in the Majorana stabilizer code. 
However, for the fermionic LDPC codes, we aim to find odd-weight logical 
operators with even overlap. 
We examine fermionic finite projective geometry codes,
fermionic finite Euclidean geometry codes, fermionic bicycle 
codes, and fermionic stacked LDPC codes, finding that these codes considered
have $2k$ odd-weight logical operators with even overlap. 
In other words, the vectors $\vec{\gamma}_1,\dots,\vec{\gamma}_{k}$ 
all have odd weight and even overlap (see Table~\ref{tab:fermion rate}). 
This finding indicates that these fermionic LDPC codes can encode $k$ logical complex
fermion modes. Thus, the code distance remains determined by the minimum weight of the 
vectors in $L$. To sum up, these fermionic LDPC codes maintain the same code parameters 
$[[n,k,d]]_{\text{f}}$ as the original qubit codes. 

\subsection{B. Examples of weakly-self-dual CSS codes}

In this subsection, we will review four types of weakly-self-dual CSS codes:
Kitaev Majorana code~\cite{kitaev2006anyons}, bicycle code~\cite{MacKay2004ieee} (Fig.~\ref{fig:graph-representation}(b)), 
Euclidean geometry codes on finite fields~\cite{cao2012novel,aly2008class} (Fig.~\ref{fig:graph-representation}(c)), and 
projective geometry codes on finite fields~\cite{farinholt2012quantum} (Fig.~\ref{fig:graph-representation}(d)).
{Very recently, we provide a method to construct weakly-self-dual CSS codes based on a CSS code that may
not be weakly self dual (see Ref.~\cite{liu2026stacked}).}

\begin{figure}[htbp]
	\centering
	\includegraphics[width=\textwidth]{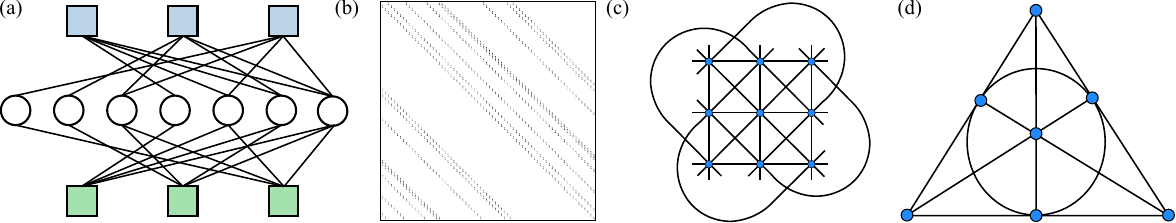}
	\caption{(a) The Tanner graph description of the Steane code. The circles, blue and green squares represent 
		the fermions, the $\gamma$-type stabilizers, and the $\gamma'$-type stabilizers, respectively. 
		(b) Illustration of a circulant matrix used to construct bicycle codes. 
		(c) Illustration of the Euclidean geometry when $(m,q)=(2,3)$. In this special case, the lines 
		are the straight lines on a torus with different directions. (d) The Fano plane, 
		which is the simplest example of a finite projective plane. 
		We will use many planes like this in code construction from finite projective geometry.
		}
	\label{fig:graph-representation}
\end{figure}

\subsubsection{Kitaev Majorana code\label{Kitaev mapping}}
According to Ref.~\cite{kitaev2006anyons}, any $[[n,k,d]]$ qubit stabilizer code 
defined by a qubit stabilizer group $\mathcal{S}$ can be mapped to a $[[2n,k,d]]_{\text{f}}$ 
Majorana stabilizer code defined by $\mathcal{S}_{\text{maj}}$. For each qubit 
$j$, four Majorana operators ${\gamma}_j^g,\ {\gamma}_j^{x,y,z}$
are introduced so that the Majorana stabilizer code is defined by $4n$ Majorana operators.
The Pauli operators of each qubit $j$ in the generators of $\mathcal{S}$ are mapped to
the product of these Majorana operators by
\begin{equation}
\begin{cases}
{\sigma}_{j}^x\rightarrow \ii{\gamma}_j^{x}{\gamma}_j^g \\
{\sigma}_{j}^y\rightarrow \ii{\gamma}_j^{y}{\gamma}_j^g\\
{\sigma}_{j}^z\rightarrow \ii{\gamma}_j^{z}{\gamma}_j^g\\
\end{cases},
\end{equation}
forming generators in $\mathcal{S}_{\text{maj}}$.
In addition, $n$ gauge stabilizers ${G}_j={\gamma}_j^{x}{\gamma}_j^{y}{\gamma}_j^{z}{\gamma}_j^g$
are added into $\mathcal{S}_{\text{maj}}$. 
The $k$ logical operators of the Majorana stabilizer code are given by replacing the 
Pauli operators in the original qubit code with the Majorana modes. 
Specifically, consider a $[[n,k,d]]$ qubit CSS code with an $X$ check matrix $H_X$
and $Z$ check matrix $H_Z$ satisfying $H_X H_Z^T = 0$.  
The corresponding Majorana stabilizer code has the check matrix  
\begin{equation}
H = \begin{pmatrix}
H_X & H_X & 0 & 0\\
H_Z & 0   & 0 & H_Z\\
I_n   & I_n   & I_n & I_n
\end{pmatrix},
\end{equation}
where $I_n$ is an $n\times n$ identity matrix.
In this case, 
the group of Majorana operators reads
$\text{Maj}(4n)\equiv \{{\Gamma}=\eta\prod_{j=1}^{n} ({\gamma}_j^g)^{\alpha_{j,1}} 
(\gamma_j^x)^{\alpha_{j,2}} (\gamma_j^y)^{\alpha_{j,3}} (\gamma_j^z)^{\alpha_{j,4}} |
\alpha_{j,1},\alpha_{j,2},\alpha_{j,3},\alpha_{j,4} \in \{0,1\}, \eta \in \{\pm 1, \pm \ii\} \}$.
The exponent factors of a Majorana stabilizer generator in $\text{Maj}(4n)$
correspond to each row in $H$. 
Since there exists a stabilizer corresponding to a vector of $(1, 1, \dots, 1)$, all
logical operators must have even weight so as to
commute with this stabilizer. In other words, one cannot find odd-weight logical operators
in the Kitaev Majorana stabilizer code.

\subsubsection{Finite Euclidean geometry codes\label{Finite Euclidean geometry code}}

In 2008, Aly proposed a class of self-dual qubit CSS codes based on 
the incidence matrix of finite Euclidean geometry~\cite{aly2008class}. 
However, we find that Aly's finite Euclidean geometry codes fail to produce odd-weight 
logical Majorana operators. We adopt the revised version of finite Euclidean geometry codes 
proposed by Cao et al. in Ref.~\cite{cao2012novel} to construct our fermionic 
LDPC code that contains odd-weight logical Majorana operators.

Here, we will follow Ref.~\cite{cao2012novel} to briefly 
review the generalized finite Euclidean geometry code. 
For positive integers \((m,q)\), a Euclidean geometry \(\text{EG}(m,q)\) consists of \(q^m\) points, each represented by an \(m\)-tuple. Here, \(m\) denotes the spatial dimension, while \(q\) represents the number of points per dimension. As shown in Fig.~\ref{fig:graph-representation}(c), \(\text{EG}(2,3)\) has two dimensions with three points per dimension, forming the Euclidean space through their Cartesian product. Consider \(q=p^s\) for some prime \(p\) and positive integer \(s\), with tuple elements drawn from the Galois field \(\mathbb{F}_{q}\). First, we describe point representation before defining lines. For any prime \(p\) and positive integer \(s\), there exists \(\alpha\) satisfying \(\alpha^{p^s}-\alpha=0\), where \(\mathbb{F}_{p^s}=\left\{0,1,\alpha,\alpha^2,\dots,\alpha^{p^s-2}\right\}\)~\cite{storme2006finite}. 
For any positive integer \(m\), \(\mathbb{F}_{q^m}\) constitutes an extension field over \(\mathbb{F}_q\), with elements expressible as \(e=a_0+a_1\alpha+\dots+a_{m-1}\alpha^{m-1}\) (\(a_i \in \mathbb{F}_q\)). As illustrated in 
Fig.~\ref{fig:graph-representation}(c), each line in \(\text{EG}(2,3)\) contains exactly three points. Thus, a bijection exists between \(\mathbb{F}_{q^m}\) and \(m\)-tuples in \(\text{EG}(m,q)\). A line in \(\text{EG}(m,q)\) (or \(\mathbb{F}_{q^m}\)) is defined as the set containing exactly \(q\) points: \(\left\{e_j + \beta e_k \,|\, \beta \in \mathbb{F}_q\right\}\), where \(e_j, e_k \in \mathbb{F}_{q^m}\) and \(e_k \neq 0\)~\cite{storme2006finite}. We now enumerate and classify these lines. There are \(q^m\) choices for \(e_j\) and \(q^m-1\) for \(e_k\). Since each line contains exactly \(q\) points, \(q\) distinct \(e_k\) choices yield identical lines. Additionally, \(\left\{e_j + \alpha \beta e_k\right\}\) with \(\beta \in \mathbb{F}_{q}\setminus\{0\}\) produces identical line sets. Consequently, the total number of lines is \(\frac{q^m(q^m-1)}{q(q-1)}\), 
while the number of lines excluding the origin point is \(\frac{(q^{m-1}-1)(q^m-1)}{q-1}\)~\cite{cao2012novel,storme2006finite}.  

For a line \(L=\left\{e_{l_0},e_{l_1},\dots,e_{l_{q-1}}\right\}\), 
translation via multiplication by a nonzero \(e \in \mathbb{F}_{q^m}\) 
yields \(eL=\{ee_{l_0},ee_{l_1},\dots,ee_{l_{q-1}}\}\). The \(q^m-1\) 
lines \(\left\{L,eL,e^2L,\dots,e^{q^m-2}L\right\}\) form an equivalence class. 
There exist \(J=\frac{q^{m-1}-1}{q-1}\) such equivalence classes~\cite{aly2008class}.
\(J(q^m-1) \times (q^m-1)\) binary check matrices \(H_j\) (\(1\leq j\leq J\))
are then constructed with rows corresponding to lines in the \(j\)-th 
equivalence class (each row is a binary vector indicating inclusion of points 
from the origin-excluding line)~\cite{cao2012novel}. 
The full check matrix is then defined as
\begin{equation}
	H_X=H_Z=\left(H_1^T,\dots,H_J^T, H_1, \dots, H_J \right).
\end{equation}
Since $H_j$ with $j=1,\dots,J$ are cyclic matrices, $H_X$ satisfies $H_X H_X^T=0$
and thus gives rise to a self-dual CSS code.

\subsubsection{Finite projective plane code\label{Finite projective geometry codes}}

The finite projective plane code is a class of self-dual CSS codes constructed by Farinholt, with the help of the finite projective plane~\cite{albert2015introduction}. 
A projective plane \(\text{PG}(2, q)\) with \(q = 2^s\) consists of a set of points and lines that satisfy:
\begin{enumerate}
\item Any two distinct points determine a unique line;
\item Any two distinct lines intersect at a unique point;  
\item There exist four points, no three of which are collinear;
\item Each line contains \(q+1\) points;
\item Each point lies on \(q+1\) lines;
\item There are exactly \(q^2 + q + 1\) points and lines~\cite{albert2015introduction}.  
\end{enumerate}
We depict $\text{PG}(2,2)$ as an example in Fig.~\ref{fig:graph-representation}(d).
The points can be represented by the equivalence classes of 
$[x,y,z]\equiv\left\{[cx, cy, cz]|[x,y,z]\neq[0,0,0],\ c\in \mathbb{F}\setminus\left\{0\right\}\right\}$. 
We note that there are ${(q^3-1)}/({q-1})=q^2+q+1$ such equivalence classes, 
which correspond to the points in the projective plane. 
The lines can be represented by the equivalence classes 
$(a,b,c)\equiv\left\{(\lambda a, \lambda b, \lambda c)|(a,b,c)\neq[0,0,0],\ \lambda\in \mathbb{F}\setminus\left\{0\right\}\right\}$, such that $ax+by+cz=0$. Each line contains $q+1$ points, and 
there are $q^2+q+1$ such lines.

The incidence matrix of the projective plane is a binary matrix \(B\) 
with its entry 
\(B_{i,j} = 1\) if point \(i\) lies on line \(j\), and \(0\) otherwise. {Specifically, we only use a subset of lines and points (i.e., skew lines and non-hyperoval points~\cite{farinholt2012quantum}) to construct matrix $B$.}
The check matrix of this self-dual CSS code is then defined as

\begin{equation}
H_X=H_Z=H= \begin{pmatrix}
B & \textbf{1}\\
\end{pmatrix},
\end{equation}
where \(\textbf{1}\) is a column vector of ones. 
We note that $HH^T=BB^T+\textbf{1}\textbf{1}^T=0$, because the overlap 
between any two rows of \(B\) is odd ($q+1$ or $1$), and thus $BB^T$ is 
a $\frac{q^2-q}{2} \times \frac{q^2-q}{2}$ binary matrix with all entries being one. 

\subsubsection{Bicycle code}\label{Mackay codes}

In 2004, Mackay et al. proposed a so-called bicycle code in Ref.~\cite{MacKay2004ieee}, 
which is a self-dual CSS code with two check matrices $H_X$ and 
$H_Z$ that satisfy $H_X=H_Z$ and $H_XH_X^T=0$, thereby 
enabling the construction of Majorana LDPC codes.
The check matrix of the bicycle code on $2n$ qubits is defined by
\begin{equation}
	H_X=H_Z = \left[C, C^T\right],  
\end{equation}
where $C$ is an $n\times n$ cyclic matrix; each row of a cyclic matrix 
is obtained by right-cyclic shifting of the previous row, 
as illustrated in Fig.~\ref{fig:graph-representation}(b). Formally, 
let $S$ be an $n\times n$ matrix defined by 
$S_{i,j}=1$ ($1\leq i,j\leq n$) if $j=(i+1)\ \text{mod}\ n$ and 
$S_{i,j}=0$ otherwise, which reads
\begin{equation}
	S=\left(\begin{matrix}
		0 & 1 & 0 & \dots & 0\\
		0 & 0 & 1 & \dots & 0\\
		\vdots & \vdots & \vdots & \dots & 0\\
		0 & 0 & 0 & \dots & 1\\
		1 & 0 & 0 & \dots & 0\\
	\end{matrix} \right).
\end{equation}
The cyclic matrix $C$ with weight $r$ is defined by 
$C=\sum_{i=1}^{r}S^{a_i}$, where $1\leq a_i \leq n$.
The code is a self-dual CSS code since $H_X H_X^T=CC^T+C^TC=0$.

As a demonstration, we construct a $[[100,20,7]]_{\text{f}}$ bicycle Majorana code
via a blind search. The same approach can be applied to search for other codes accordingly.
In addition, we find that the other three types of codes proposed by Mackay et al. 
in Ref.~\cite{MacKay2004ieee} are not suitable for the fermionic LDPC code since
one cannot find odd-weight logical Majorana operators similar to 
the Kitaev Majorana code.

\subsection{C. Extracting odd-weight logical Majorana operators\label{Fermionic stabilizer codes}}

We now describe how to identify odd-weight fermionic logical operators from the
Majorana check matrix $H$ constructed from check matrices $H_X=H_Z$ of 
a self-dual $[[n,k_{\text{q}},d]]$ CSS code, that is,
\begin{equation}
	H=\begin{pmatrix}
		H_\gamma & 0 \\
		0        & H_{\gamma^\prime}
	\end{pmatrix}
\end{equation}
where $H_\gamma=H_{\gamma^\prime}=H_X=H_Z$. $H$ is an $(n-k_{\text{q}})\times 2n$ binary matrix 
with all row vectors linearly independent such that $\text{rank}(H)=n-k_{\text{q}}$. 
All logical operators correspond to nontrivial vectors in $\ker(H)/\text{im}(H^T)$.
Let $S_{\text{L}}=\left\{\vec{v}_j|j=1,\dots,2k_{\text{q}} \right\}$ be a basis for $\ker{H}/\text{im}(H^T)$,
where the weight of $\vec{v}_j$ can be either even or odd~\cite{roffe2019quantum}. 
An orthonormal set $\left\{\vec{\gamma}_j \right\}_{j=1}^{2k_{\text{f}}}$ with $k_{\text{f}}\leq k_{\text{q}}$ 
consisting of the odd-weight vectors will be constructed from $S_{\text{L}}$.
We show the procedure to construct the set constituting a basis for our code space of 
fermions in the following, generalizing the method in Ref.~\cite{prasetiaextended} to
the $\mathbb{F}_2$ case. 

If vectors in $S_{\text{L}}$ all have even weight, then there 
do not exist odd-weight logical operators, since the addition of any two even-weight
binary vectors yields an even-weight vector.
Otherwise, there exists at least one odd-weight vector, say, $\vec{v}_1$. 
We choose $\vec{\gamma}_1=\vec{v}_1$ as the first element in the odd-weight 
vector set and update the remaining $2k_{\text{q}}-1$ vectors in $S_{\text{L}}$ 
as $\vec{v}_j^\prime\to\vec{v}_j-\left( \vec{v}_j\cdot\vec{v}_1\right)\vec{v}_1$.
As a result, $\vec{v}_1$ is orthogonal to all the remaining vectors $\vec{v}_j^\prime$ ($j\geq 2$).
We repeat this process until the remaining vectors are all of even weight,
leading to an orthonormal set consisting of odd-weight vectors. 
The algorithm is summarized in Algorithm 1. We note that this algorithm 
is similar to the algorithm in Ref.~\cite{tansuwannont2025clifford}, which is aimed at 
constructing compatible symplectic basis from a weakly-self-dual CSS code.

\begin{algorithm}[H]
	\label{alg:Gram-Schmidt}
	\caption{Orthonormalization of $\left\{\vec{v}_j \right\}_{j=1}^{2k_{\text{q}}}$ over $\mathbb{F}_2$}
	\SetAlgoLined
	\KwIn{A set of binary vectors $\{\vec{v}_1, \vec{v}_2, \dots, \vec{v}_{2k_{\text{q}}}\}$ over $\mathbb{F}_2$}
	\KwOut{The orthonormal basis of odd weight operators $\left\{\vec{\gamma}_j \right\}_{j=1}^{2k_{\text{f}}}$}
	
	\SetKwFunction{FMain}{Gram-Schmidt orthonormalization over $\mathbb{F}_2$}
	\SetKwProg{Fn}{Function}{:}{}
	
	\Fn{\FMain{$\{\vec{v}_1, \dots, \vec{v}_{2k_{\text{q}}}\}$}}{
		$\mathcal{B} \gets \{\vec{v}_1, \dots, \vec{v}_{2k_{\text{q}}}\}$ \tcp*{Copy input vectors}
		$\mathcal{O} \gets \emptyset$ \tcp*{Initialize orthogonal basis set}
		
		\While{$\exists \vec{b}_i \in \mathcal{B} \text{ such that } \vec{b}_i \cdot \vec{b}_i = 1$}{
			$\mathcal{O} \gets \mathcal{O} \cup \left\{\vec{b}_i \right\}$\;
			$\mathcal{B} \gets \mathcal{B} \setminus \left\{\vec{b}_i \right\}$\;
			\ForEach{$\vec{b} \in \mathcal{B}$}{
				$\vec{b} \gets \vec{b} - (\vec{b} \cdot \vec{b}_i)\vec{b}_i$\;
			}
		}
		\Return $\mathcal{O}$\;
	}
\end{algorithm}

The following theorem provides a necessary and sufficient condition for 
the existence of an orthonormal basis for $\mathrm{span}(S_{\text{L}})$,
i.e., an orthonormal basis consisting of odd-weight vectors.

\begin{theorem}
Let $G$ be a $2k_{\text{q}} \times 2k_{\text{q}} $ binary matrix defined by 
$G_{ij} = \vec{v}_i \cdot \vec{v}_j$ with $1 \le i,j \le 2k_{\text{q}}$.
An orthonormal basis for $\mathrm{span}(S_{\text{L}})$ 
exists if and only if there exists an invertible matrix $P$ such that
\begin{equation}
	P G P^T = I_{2k_{\text{q}}}.
\end{equation} 
\end{theorem}

\begin{proof}
	\label{Gram matrix}
	If there exists an orthonormal basis $\{\vec{\gamma}_1,\dots,\vec{\gamma}_{2k_{\text{q}}}\}$ 
	for the span, then we can write each vector in the basis as a linear combination of
	the vectors $\vec{v}_1, \dots, \vec{v}_{2k_{\text{q}}}$, that is, there exists a $2k_{\text{q}} \times 2k_{\text{q}}$
	matrix $P$ such that 
	\begin{equation}
		\begin{pmatrix} \vec{\gamma}_1 \\ \vdots \\ \vec{\gamma}_{2k_{\text{q}}} \end{pmatrix} 
		= P \begin{pmatrix} \vec{v}_1 \\ \vdots \\ \vec{v}_{2k_{\text{q}}} \end{pmatrix}.
	\end{equation}
    Since $\{\vec{\gamma}_1,\dots,\vec{\gamma}_{2k_{\text{q}}}\}$ is orthonormal, we have
    $\vec{\gamma}_i \cdot \vec{\gamma}_j=\delta_{i,j}=
    \sum_{i^\prime, j^\prime}P_{i i^\prime} P_{j j^\prime}\vec{v}_{i^\prime}\cdot \vec{v}_{j^\prime}
    =[P G P^T]_{ij}$, that is, $P G P^T=I_{2k_{\text{q}}}$.
    
    Conversely, if there exists an invertible matrix $P$ such that $P G P^T = I_{2k_{\text{q}}}$,
    we define a new basis $\{\vec{\gamma}_1,\dots,\vec{\gamma}_{2k_{\text{q}}}\}$  by
	\begin{equation}
		\begin{pmatrix} \vec{\gamma}_1 \\ \vdots \\ \vec{\gamma}_{2k_{\text{q}}} \end{pmatrix} 
		= P \begin{pmatrix} \vec{v}_1 \\ \vdots \\ \vec{v}_{2k_{\text{q}}} \end{pmatrix},
	\end{equation}
	which clearly gives $\vec{\gamma}_i \cdot \vec{\gamma}_j=\delta_{i,j}$.
	Thus, $\{\vec{\gamma}_1,\dots,\vec{\gamma}_{2k_{\text{q}}}\}$ forms an orthonormal basis.
\end{proof}

We numerically calculate the number of odd-weight logical Majorana operators forming $k_{\text{f}}$ logical complex fermion modes for the check matrix $H$ constructed based on bicycle codes, finite Euclidean geometry codes, and finite projective plane codes, and compare them with the number of logical qubits $k_{\text{q}}$ for the corresponding self-dual CSS codes in Table~\ref{tab:fermion rate}.
We see that for all cases, the algorithm can efficiently find the same number of logical complex fermion modes as that of logical qubits, thus providing a high encoding rate.

\begin{table}[htbp]
	\centering
	\caption{Comparison between the number of logical qubits and the number of logical complex fermion modes for different codes. 
		$d$ denotes the distance of the self-dual CSS codes.}
	\label{tab:fermion rate}
	\begin{tabular}{c|c|c|c|c|c|c|c|c|c}
		\hline
		Name & Bicycle code & $\text{EG}(2,4)$ & $\text{EG}(2,8)$& $\text{EG}(3,2)$&$\text{EG}(3,4)$&$\text{PG}(2,4)$ & $\text{PG}(2,8)$ & $\text{PG}(2,16)$&$\text{PG}(2,32)$\\
		\hline
		\(n\) & 100 & 30 & 126 & 42 & 630 & 16 & 64 & 256 & 1024\\
		\(k_{\text{q}}\) & 20  & 2  & 38  & 30 & 506 & 6 & 44 & 190 & 812\\
		\(k_{\text{f}}\) & 20  & 2  & 38  & 30 & 506 & 6 & 44 & 190 & 812\\
		\(d\) & 7   & 5  & 9  & 2 & 5 & 3 & 5 &9  & 17\\
		\hline
	\end{tabular}
\end{table}

\section{S-2. Fermionic state teleportation}

We now validate the quantum state teleportation circuit shown in Fig.~3(a) in the main text.
Suppose that our fermionic LDPC code memory $A$ encodes $k_{\text{f}}$ logical complex fermion modes, 
and we want to transfer the logical information of the $j$th logical mode to a processor $B$,
which is initialized in $|\overline{0}\rangle_B$.
Let the logical state of $A$ be $\ket{\overline{\psi}}_A$, 
and the state of the entire system $\ket{\overline{\Psi}}$ can be written as 
\begin{equation} 
	\ket{\overline{\Psi}}=\ket{\overline{\psi}}_A\ket{\overline{0}}_B=\left(\overline{f}+\overline{c}_j^\dagger \overline{g}\right)\ket{\text{vac}},
\end{equation}
where $\overline{c}_{j}^\dagger$ denotes the logical fermionic creation operator of the $j$th logical mode of code $A$, 
$\overline{h} =h 
(\overline{c}_1^\dagger,\dots,\overline{c}_{j-1}^\dagger,
\overline{c}_{j+1}^\dagger,\dots, \overline{c}_{k_{\text{f}}}^\dagger)$ 
with $h=f,g$
represents a polynomial of the logical operators except the $j$th operator, 
and $\ket{\text{vac}}$ is the vacuum state.
We prove that the logical state of the system will be 
$\ket{\overline{\Psi}^\prime}=\left(\overline{f}+\overline{c}_B^\dagger \overline{g}\right)\ket{\text{vac}}$ 
after implementing the quantum state teleportation circuit, 
where $\overline{c}_B^\dagger$ is the encoded logical creation operator of code $B$.

Let $\overline{c}_{j}^\dagger\equiv\frac{1}{2}\left(\overline{\gamma}_{j}+\ii\overline{\gamma}_{j}^\prime\right)$ 
and $\overline{c}_{B}^\dagger\equiv\frac{1}{2}\left(\overline{\gamma}_{B}+\ii\overline{\gamma}_{B}^\prime\right)$, 
the first step of the circuit is to perform the 
projective measurement of the joint parity operator 
$\ii\overline{\gamma}_j\overline{\gamma}_B$, 
resulting in the state proportional to
\begin{align}
	\ket{\overline{\Psi}_1}&\sim\frac{1\pm \ii\overline{\gamma}_j\overline{\gamma}_B}{2}\ket{\overline{\Psi}}\nonumber\\
	&\sim \left(\overline{f}+\overline{c}_j^\dagger \overline{g}\pm \ii\overline{\gamma}_j\overline{\gamma}_B\overline{f}\pm \ii\overline{\gamma}_j\overline{\gamma}_B\overline{c}_j^\dagger \overline{g}\right)\ket{\text{vac}}\nonumber \\
	&\sim \left(\overline{f}+\overline{c}_j^\dagger \overline{g}\pm \ii\overline{\gamma}_j\overline{\gamma}_B\overline{f}\mp \ii\overline{\gamma}_B \overline{g}\right)\ket{\text{vac}},
\end{align}
where the last equation is due to the fact that 
$\overline{\gamma}_j=\overline{c}_j^\dagger+\overline{c}_j$. 
If the measurement result is $+1$,
we proceed;
otherwise, a $\text{Z}$ gate
$\text{Z}=\exp\left(\ii \pi\overline{c}^\dagger_B\overline{c}_B\right)$ 
is applied to the processor to derive the state
\begin{align}
	\ket{\overline{\Psi}_2}&=\exp\left(\ii \pi\overline{c}^\dagger_B\overline{c}_B\right)\ket{\overline{\Psi}_1}\nonumber\\
	&= \left(1-2\overline{c}^\dagger_B\overline{c}_B\right)\ket{\overline{\Psi}_1}\nonumber \\
	&\sim \left(\overline{f}+\overline{c}_j^\dagger \overline{g}+ \ii\overline{\gamma}_j\overline{\gamma}_B\overline{f}- \ii\overline{\gamma}_B \overline{g}\right)\ket{\text{vac}}.
\end{align}
We now perform the projective measurement of the particle number operator $\overline{c}^\dagger_j\overline{c}_j$, 
and the state becomes
\begin{align}
	\ket{\overline{\Psi}_3}&\sim\begin{cases}
		\left(1-\overline{c}^\dagger_j\overline{c}_j\right)\ket{\overline{\Psi}_2}\\
		\overline{c}^\dagger_j\overline{c}_j\ket{\overline{\Psi}_2}
	\end{cases}\nonumber\\
	&= \begin{cases}
		\left(\overline{f}- \ii\overline{\gamma}_B \overline{g}\right)\ket{\text{vac}}\\
		\left(\ii\overline{\gamma}_j\overline{\gamma}_B\overline{f}+\overline{c}_j^\dagger \overline{g}\right)\ket{\text{vac}}
	\end{cases}\nonumber,
\end{align}
for the measurement result of $0$ and $1$, respectively. 
If the measurement result is $0$, 
we simply proceed;
otherwise, we apply the joint operator 
$\ii\overline{\gamma}_j\overline{\gamma}_B$ to get the state
$\ket{\overline{\Psi}_4}=\left(\overline{f}- \ii\overline{\gamma}_B \overline{g}\right)\ket{\text{vac}}=\left(\overline{f}- \ii\overline{c}_B^\dagger \overline{g}\right)\ket{\text{vac}}$. 
After that, we apply a phase rotation gate 
$\overline{S}=\exp\left(\ii \frac{\pi}{2}\overline{c}^\dagger_B\overline{c}_B\right)$
to achieve the target state 
$\ket{\overline{\Psi}^\prime}=\left(\overline{f}+\overline{c}_B^\dagger \overline{g}\right)\ket{\text{vac}}$.

{One can easily derive that 
$\ii\overline{\gamma}_j\overline{\gamma}_B=-\overline{Z_j}\,(\overline{\text{Braid}(j,B)})^2$,
where $\overline{\text{Braid}(j,B)}=\frac{1}{\sqrt{2}} (1-\overline{\gamma}_j' \overline{\gamma}_B)$
is the logical braiding gate~\cite{Schuckert2024arxiv}
and $\overline{Z_j}=\ii \overline{\gamma}_j \overline{\gamma}_j'$.
Therefore, $\ii\overline{\gamma}_j\overline{\gamma}_B$ can be
realized by applying two logical braiding gates and one logical $Z$ gate as shown in Fig.~3(a),
which can be realized transversally
as shown in Fig.~\ref{Teleportation}(b).}

For two fermion gates, we need to transfer the information 
from two logical fermionic sites in the memory to two processors, 
and apply corresponding logical gates on the processors.

\begin{figure}
\centering 
\includegraphics[width=0.4\textwidth]{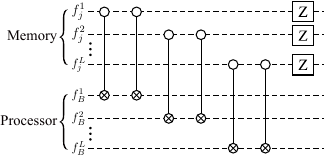}
\caption{
{Illustration of implementing \(\ii\overline{\gamma}_j\overline{\gamma}_B\) transversally 
between a memory and a processor. $f_j^1, f_j^2, \dots, f_j^L$ denote the physical fermion 
sites in the memory so that the logical Majorana operators $\overline{\gamma}_j$ and 
$\overline{\gamma}_j'$ are given by product of physical Majorana operators at these
sites (similarly for the processor).}
} 
\label{Teleportation}
\end{figure}

\section{S-3. Fermionic lattice surgery}

In the main text, we introduce two methods for fermionic lattice surgery. 
In this section, we first explain why qubit-based lattice surgery cannot be directly generalized to fermionic codes. 
We then provide details of our two constructions, analyze their resource overhead, and demonstrate their fault tolerance by proving their code-distance preservation.

\subsection{A. Qubit lattice surgery cannot be directly generalized}
\label{sec:qubit-versus-fermionic-lattice-surgery}

We first explain why existing lattice-surgery constructions for qubit
error-correcting codes cannot be directly applied to fermionic
stabilizer codes, so that fermionic lattice surgery has to be designed independently to ensure that all newly introduced stabilizers have even weight, satisfying all commutation conditions. 

A qubit CSS code is specified by two parity-check matrices $H_X$ and $H_Z$, which are required to satisfy $H_XH_Z^T=0$. 
This condition guarantees that every X-type stabilizer commutes with
every Z-type stabilizer. 
Stabilizers of the same Pauli type commute automatically, regardless of the overlap between their supports.
Consequently, a qubit lattice-surgery construction only needs to preserve the mixed commutation relation in $H_XH_Z^T=0$ when new ancilla qubits and merged
stabilizers are introduced.
The situation is different for fermionic stabilizer codes. 
Majorana strings of the same type do not commute automatically. 
In fact, the check matrices of a valid fermionic stabilizer
code must satisfy $H_{\gamma}H_{\gamma}^T=H_{\gamma'}H_{\gamma'}^T=0$.
Moreover, every physical Majorana stabilizer have even weight. 
These requirements are absent in the qubit CSS setting.

As a concrete example, consider the generalized lattice-surgery
construction of Ref.~\cite{cohen2022low} for measuring the joint logical observable $X_1 Z_2$ between two quantum LDPC code blocks. 
In this construction, the two ancilla systems are connected by mixed XZ-type stabilizers. 
All the resulting Pauli stabilizers commute, and their product gives the target logical observable $X_1 Z_2$.
Suppose that we retain exactly the same stabilizer supports and directly
replace $X$ and $Z$ with Majorana operators $\gamma$ and $\gamma'$, respectively. 
The boundary checks at the free ends of the two ancilla systems, the checks connecting the two ancilla systems, and the checks attaching them to the two code blocks then include odd-weight Majorana strings.
The same anticommutation occurs between the other disjoint odd-weight
boundary checks introduced by the direct replacement. 

To analyze this problem more systematically, we reformulate the qubit-based lattice surgery in the language of homological algebra~\cite{ide2025fault}. 
Consider CSS codes before lattice surgery, described by the following chain complex:
\begin{equation}
C:\quad
C_2\xrightarrow{H_X^T}C_1\xrightarrow{H_Z}C_0,
\qquad
H_ZH_X^T=0,
\end{equation}
where $C_1$ is the physical-qubit space, while $C_2$ and $C_0$
correspond to the X-check space and the Z-syndrome space,
respectively. 
When ancilla qubits are introduced for lattice surgery, we augment this code with an ancilla code complex
\begin{equation}
A:\quad
A_1\xrightarrow{\partial_1}A_0
\xrightarrow{\partial_0}A_{-1}.
\end{equation}
The merged code is then described by the cone code
\begin{equation}
\begin{tikzcd}[
    column sep=4.5em,
    row sep=1.2em
]
A_1
    \arrow[r, "\partial_1"]
    \arrow[ddr, "f_1"']
&
A_0
    \arrow[r, "\partial_0"]
    \arrow[ddr, "f_0"']
&
A_{-1}
\\
\oplus
&
\oplus
&
\oplus
\\
C_2
    \arrow[r, "H_X^T"']
&
C_1
    \arrow[r, "H_Z"']
&
C_0
\end{tikzcd}
\end{equation}
where $f_1$ and $f_0$ are maps from $A_1$ to $C_1$ and from $A_0$ to $C_0$, respectively, satisfying
\begin{equation}
H_Zf_1=f_0\partial_1.
\end{equation}
This cone code has check matrices
\begin{equation}
\widetilde H_X=
\begin{pmatrix}
H_X & 0\\
f_1^T & \partial_1^T
\end{pmatrix},
\qquad
\widetilde H_Z=
\begin{pmatrix}
H_Z & f_0\\
0 & \partial_0
\end{pmatrix}.
\end{equation}
The condition $H_Zf_1=f_0\partial_1$ ensures that $\widetilde H_X\widetilde H_Z^T=0$, yielding a well-defined CSS code, while adjusting $f_0$, $f_1$, $\partial_0$, $\partial_1$ can make the target logical operator as stabilizers of the new CSS code.
The newly added X stabilizers are given by the rows of 
$\begin{pmatrix}
	f_1^T & \partial_1^T
\end{pmatrix}$.
Typically, one requires that the rows of $\partial_1^T$ sum to zero so that the product of all 
the newly added X stabilizers has no support on the ancilla qubits. Let 
$a^T=\begin{pmatrix}
	1 & 1 & \dots & 1
\end{pmatrix}$;
this condition reads $a^T \partial_1^T=\sum(\textrm{rows of }\partial_1^T )=0$, 
equivalently $\partial_1 a=0$, so that $a\in \ker(\partial_1)$. The measurement logical operator is then 
given by $a^T f_1^T$, i.e., the sum of all rows of $f_1^T$. More generally, the logical operator 
may correspond to a sum over only a subset of rows of 
$\begin{pmatrix}
	f_1^T & \partial_1^T
\end{pmatrix}$.
In that case, we still require $a\in \ker(\partial_1)$ (with $a$ now representing the chosen subset) to 
ensure $a^T \partial_1^T=0$, so the logical operator is again of the form $a^T f_1^T$.

This mapping-cone construction, however, does not directly extend to
fermionic codes.
The condition $H_Zf_1=f_0\partial_1$ guarantees only the Pauli commutation relation
$\widetilde H_X\widetilde H_Z^T=0$.
Here we use the correspondence $\widetilde H_X=\widetilde H_{\gamma}$,
$\widetilde H_Z=\widetilde H_{\gamma'}$, $H_X=H_{\gamma}$, and
$H_Z=H_{\gamma'}$.
For a merged fermionic code obtained via fermionic lattice surgery, we must instead impose
\begin{equation}
\widetilde H_X\widetilde H_X^T=0,
\qquad
\widetilde H_Z\widetilde H_Z^T=0.
\end{equation}
The mixed condition $\widetilde H_X\widetilde H_Z^T=0$ is generally not required for the fermionic code
unless this condition is equivalent to $\widetilde H_X\widetilde H_X^T=0$ if $\widetilde H_X=\widetilde H_Z$.
Consequently, a merged qubit CSS code guaranteed by $H_Zf_1=f_0\partial_1$ is generally not a valid Majorana stabilizer code, since the Majorana constraints 
require
\begin{equation} \label{Eq:fLS}
\begin{aligned}
H_Xf_1 &= 0, \qquad
f_1^Tf_1+\partial_1^T\partial_1 = 0,\\
f_0f_0^T &= 0, \qquad
f_0\partial_0^T = 0, \qquad
\partial_0\partial_0^T = 0.
\end{aligned}
\end{equation}
It remains an open question whether Eq.~(\ref{Eq:fLS}) can be used to devise more efficient fermionic lattice surgery,
a direction that merits further study.

To sum up, fermionic lattice surgery must satisfy a set of additional constraints beyond those of qubit-based surgery, rendering its systematic construction highly nontrivial.

\subsection{B. Method 1}

For the first method, the two logical operators have equal weight, 
i.e., $|\overline{\gamma}_A|=|\overline{\gamma}_B|$.
We now show that the merged code obtained using this method is a well-defined fermionic LDPC code.

First, all the new stabilizers have even weight by direct observation.
We now analyze the weights of the stabilizers and show commutativity between them.
\begin{itemize}
	\item The modified stabilizer $c_{A,i}^{\prime}$ has weight $|c_{A,i}| + m_i$, 
	which is even and remains bounded. It commutes with:
	\begin{itemize}
		\item Any measurement stabilizer $M_n$, 
		since their overlap is either empty or 
		$\left\{ \gamma_{A,n}, \gamma_{a,i,n}^{(\prime)}\right\}$,
		\item All gauge stabilizers, since their overlap is either empty or
		$\left\{ \gamma_{a,i, j_{2n-1}}, \gamma_{a,i, j_{2n}}^\prime\right\}$ for some $n$,
		\item All stabilizers and modified stabilizers in $B$, as their supports are disjoint.
		\item All stabilizers in $A$, which is ensured by the construction of code $A$.
	\end{itemize}
	
	\item The stabilizer $ c_{B,i}^{\prime}$ has a bounded even weight by definition,
	and commutes with measurement and gauge stabilizers by similar reasons for $ c_{A,i}^{\prime}$.
	
	\item The gauge stabilizers have weight four by definition, and commute with measurement stabilizers 
	since their overlap contains either zero or two Majorana modes.
	
	\item Each measurement stabilizer $M_i$ has bounded weight, since the numbers of 
	the modified stabilizers associated with $\gamma_{A,i}$ and $\gamma_{B,i}$ 
	are bounded, and thus contribute a bounded number of ancilla modes to $M_i$.
\end{itemize}
Thus, code $\mathcal{C}_{\text{merged}}$ is a well-defined fermionic LDPC code 
with mutually commuting stabilizers.
Furthermore, the number of ancilla modes is bounded by 
$4(d-1)+\sum_{i=1}^{N_A}4m_i$, 
resulting in a low overhead for the lattice surgery.
We leave the analysis of capacity of fault-tolerance to the next section.

\subsection{C. Method 2}
This method is applicable to both cases with
$|\overline{\gamma}_A|=|\overline{\gamma}_B|$ and 
$|\overline{\gamma}_A|\neq|\overline{\gamma}_B|$
and when code $B$ is either a fermionic color code
or a fermionic LDPC code. Pseudo-code is given in Alg.~\ref{alg:LatticeSurgeryTwo}
for this method.

We now discuss how to identify gauge stabilizers from a graph theory perspective, 
as illustrated in Fig.~\ref{figs-cycle}. A graph is constructed as follows:
\begin{itemize}
	\item The modified stabilizers (pink and cyan circles), measurement 
	stabilizers (yellow circles) and ancilla Majorana operators (blue circles) are treated as vertices. 
	An edge between a measurement stabilizer vertex and an ancilla mode vertex is added 
	if the corresponding measurement stabilizer contains the corresponding ancilla Majorana operator
    (see Fig.~\ref{figs-cycle}(a)).	
	\item If a modified stabilizer connects to more than one pair of ancilla modes, 
	then we add more modified stabilizer vertices and adjust connections so that 
	each modified stabilizer vertex connects to one pair of ancilla modes (see Fig.~\ref{figs-cycle}(b)). 
	\item For each pair of ancilla Majorana operators that are introduced to 
	ensure even-weight measurement stabilizers (e.g., node 3 and 15 in Fig.~\ref{figs-cycle}),
	we add an endpoint vertex between them and connect this vertex to the two 
	ancilla mode vertices (see Fig.~\ref{figs-cycle}(b)). 
	\item We refer to the vertices corresponding to measurement stabilizers as measurement vertices,
	the vertices corresponding to ancilla Majorana operators as ancilla vertices, 
	and the modified stabilizers and the vertices added in the previous two steps as modified vertices.
\end{itemize}
Therefore, this graph contains $D_{AB}=\frac{D_A+D_B}{2}$ measurement vertices,
$\sum_{i=1}^{N_A} m_{A,i} + \sum_{i=1}^{N_B} m_{B,i} + N_P$ ancilla vertices,
$\sum_{i=1}^{N_A} \frac{1}{2}m_{A,i} + \sum_{i=1}^{N_B} \frac{1}{2} m_{B,i} + \frac{N_P}{2}$ modified vertices, 
and $2(\sum_{i=1}^{N_A} m_{A,i} + \sum_{i=1}^{N_B} m_{B,i} + N_P)$ edges, since each ancilla vertex 
connects two vertices.
Since each measurement vertex has even degree, each modified vertex has degree two,
and any two vertices can be connected by some path,
this graph is an undirected Eulerian graph with a single connected component,
indicating that it contains 
\begin{equation}
	|E| - |V| + 1 = \sum_{i=1}^{N_A} m_{A,i}/2 + \sum_{i=1}^{N_B} m_{B,i}/2 + N_P/2-D_{AB} + 1
\end{equation}
independent cycles~\cite{gross2018graph}. Here $|E|$ and $|V|$ denote the number of
edges and vertices, respectively. The gauge stabilizer is then defined as 
the product of all ancilla Majorana operators on each independent cycle 
(see Fig.~\ref{figs-cycle}(c) and (d)).

The commutativity between measurement stabilizers and modified 
stabilizers is ensured by construction.
We now verify that the gauge stabilizers are even-weight operators,
and commute with measurement and modified stabilizers.

\begin{itemize}
	\item Each path consists of a set of ancilla vertices, which are interlaced with
	measurement vertices and modified vertices. We note that each modified vertex 
	connects exclusively to a pair of ancilla vertices on a cycle and each ancilla vertex connects 
	one modified vertex, as depicted in Fig.~\ref{figs-cycle}(c). Therefore, the number of ancilla vertices
	in each cycle is twice the number of modified vertices, 
	ensuring that gauge stabilizers are even-weight operators.
	\item The number of overlapped ancilla vertices between a gauge 
	stabilizer and a modified stabilizer is even, and
	thus each gauge stabilizer commutes with all modified stabilizers.
	\item The number of ancilla vertices shared by a gauge stabilizer 
	and a measurement stabilizer is always even. 
	This is because, within the cycle of a gauge stabilizer, traversing a measurement vertex requires entering and exiting through two distinct ancilla vertices, 
	as shown in Fig.~\ref{figs-cycle}(c) and (d).
	\item The gauge stabilizers commute with each other. If two cycles intersect at a vertex $\gamma_{a(b),i,j_1}$,
	they must pass through the modified vertex associated with $\gamma_{a(b),i,j_1}$. Consequently, these
	cycles must also intersect at the other vertex $\gamma_{a(b),i,j_2}^\prime$ associated with this modified vertex,
	resulting in an even overlap between the two gauge stabilizers. 
\end{itemize}

\begin{figure}
	\centering 
	\includegraphics[width=\textwidth]{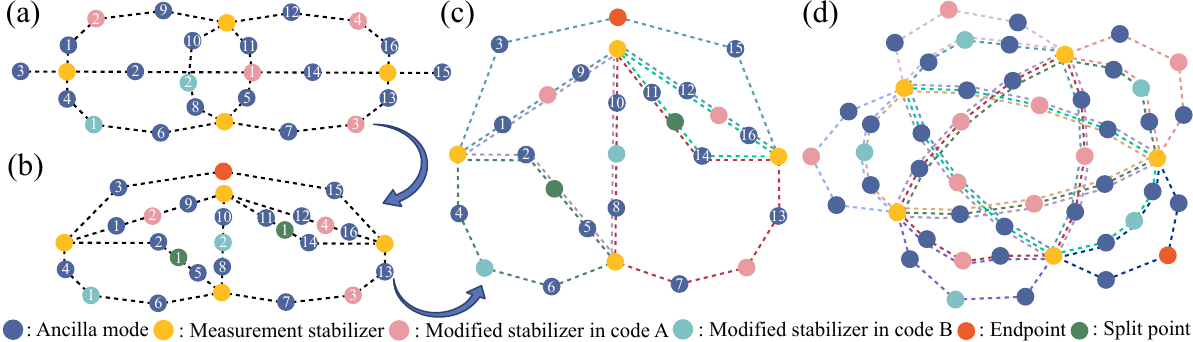}
	\caption{Illustration of construction of gauge stabilizer generators 
		in a general fermionic lattice surgery using method 2. (a) An initial graph. 
		(b) A graph obtained by splitting modified stabilizers via adding more vertices,
		called split points, and by
		adding an endpoint that connects the ancilla Majorana operators introduced to 
		ensure even-weight measurement stabilizers (node 3 and node 15). 
		(c) Independent cycles, each represented by different colored closed paths. 
		These figures correspond to the case shown in Fig.~6 in the main text.
		(d) Independent cycles for the case of performing fermionic lattice 
		surgery between $\text{PG}(2,8)$ and $d=5$ fermionic color code.} 
	\label{figs-cycle}
\end{figure}

We now demonstrate that the merged code $\mathcal{C}_{\text{merged}}$ is a
fermionic LDPC code with a low overhead. 
The number of ancilla fermions introduced is $\sum_i m_{A,i} + \sum_i m_{B,i} + N_P$, 
which is bounded by $\sum_i m_{A,i} + \sum_i m_{B,i} + D_{AB}$. 
This corresponds to a low-overhead construction, 
making this lattice surgery practical.
The weight of measurement stabilizers 
is bounded due to the limited number of original stabilizers associated with each site.
For the weight of gauge stabilizers,
we numerically compute the maximum weight of the gauge stabilizers for 1000 
randomly generated code pairs ($A$ and $B$). Our results show that this weight 
scales as $[(|\overline{\gamma}_A|+|\overline{\gamma}_B|)/2]^{0.5}$, 
as illustrated in Fig.~\ref{figs-weight}.

\begin{figure}[htbp]
	\centering 
	\includegraphics[width=0.35\textwidth]{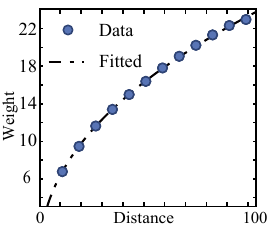}
	\caption{Numerically computed maximum weight of the gauge stabilizers 
		for 1000 randomly generated code pairs ($A$ and $B$) with respect to the average weight of 
		$\overline\gamma_A$ and $\overline\gamma_B$.
		The dotted-dashed line is power-law fit of the data. } 
	\label{figs-weight}
\end{figure}

The fault-tolerant measurement of the joint logical operator 
$\ii\overline{\gamma}_A\overline{\gamma}_B$ is due to the fact 
that this joint operator becomes a stabilizer of the merged code.
To show this, we need to demonstrate that there exists a subset of gauge stabilizers,
such that the product of these gauge stabilizers incorporates all ancilla Majorana operators,
since the product of all measurement stabilizers contains $\overline{\gamma}_A\overline{\gamma}_B$
and all ancilla Majorana operators.
The product of all ancilla modes corresponds to a big cycle containing all ancilla vertices,
which can always be decomposed into a subset of the independent cycles found above.
As a result, the product of all measurement stabilizers and this subset of gauge stabilizers
leads to $\ii\overline{\gamma}_A\overline{\gamma}_B$.
The measurement result of this operator is determined by the product of the measurement 
results of all measurement stabilizers and this subset of gauge stabilizers,
In fact, only measurement results of measurement stabilizers are necessary since 
they are random due to their anticommutation with some original stabilizers in blocks 
$A$ and $B$, while the gauge stabilizer results
are always $+1$, since they belong to the stabilizer group of the original code.

In Fig.~\ref{figs-cycle}(d), we show an example depicting the results of 
independent cycle search for performing 
fermionic lattice surgery between $\text{PG}(2,8)$ and a $d=5$ fermionic color code.

\begin{algorithm}[H]
	\caption{Fermionic Lattice Surgery}
	\label{alg:LatticeSurgeryTwo}
	\SetAlgoLined
	\KwIn{
		Codes $A$, $B$ with stabilizers $\mathcal{S}_A$, $\mathcal{S}_B$, Logical fermionic site supports $F_A$, $F_B$
	}
	\KwOut{Merged code $\mathcal{C}_{\text{merged}}$}

	$\mathcal{A}$ (ancilla mode set) $\gets \emptyset, \mathcal{S}_{\text{fix}} \gets \{S \in \mathcal{S}_A \cup \mathcal{S}_B | \text{supp}(S) \cap (F_A \cup F_B) \neq \emptyset\}, \mathcal{S}_{\text{free}} \gets (\mathcal{S}_A \cup \mathcal{S}_B) \setminus \mathcal{S}_{\text{fix}}$ \\
	\For{$S \in \mathcal{S}_{\text{fix}}$}{
		\For{$i=1$ \KwTo $\lfloor\frac{1}{2} |\text{supp}(S)\cap(F_A \cup F_B)| \rfloor$}{
			introduce new ancilla site $a$, $\mathcal{A} \gets \mathcal{A} \cup \{a\}, S \gets S \ii\gamma_a \gamma_a^\prime$\\
		}
	}
	$\mathcal{M}$ (measurement stabilizers) $\gets \emptyset$ \\
	\ForEach{$i \in 1:\min(|F_A|,|F_B|)$}{
		$\text{flag}\gets \text{False}, M_i \gets\ii\gamma_{F_A[i]}\gamma_{F_B[i]}$, Attach relevant ancillas to $M_i$ without reuse\\
		\If{$\text{wt}(M_i)$ is odd}{
			\eIf{$\text{flag}$ is $\text{False}$}{
				introduce new ancilla site $a$, $\mathcal{A} \gets \mathcal{A} \cup \{a\}, M_i\gets M_i\gamma_a, \text{flag}\gets \text{True}$ \\
			}
			{
				$M_i\gets M_i\gamma_a', \text{flag}\gets\text{False}$\\
			}
		}
		$\mathcal{M} \gets \mathcal{M} \cup \{M_i\}$\\
	}
	Add pair operators for unpaired fermions\\
	$G \gets \text{graph}(\mathcal{A})$, add edges for co-occurring ancillas in $\mathcal{S}_{\text{fix}} \cup \mathcal{M}$ \\
	$\mathcal{G}$ (gauge stabilizers) $\gets \{\prod_{v\in C} \gamma_v | C \in \text{Independent cycles}(G)\}$ \\
	\Return $\mathcal{C}(\mathcal{S}_{\text{free}} \cup \mathcal{S}_{\text{fix}} \cup \mathcal{M} \cup \mathcal{G})$
\end{algorithm}

\subsection{D. Fault-tolerance of fermionic lattice surgery}

In this subsection, we will analyze the fault-tolerance of fermionic lattice surgery
using the subsystem code formalism~\cite{vuillot2019code,aly2006subsystem}. 
The procedure's fault-tolerance is determined 
by the distance of the dressed logical operators, which will be defined shortly.
For method 2 with a processor being a color code, we show that the distance of the dressed 
logical operators is at least the minimum weight of the original logical operators. 
For method 1, we provide numerical 
evidence demonstrating that the dressed logical distance is no smaller 
than the original logical distance.

A subsystem code is defined on $2n$ Majorana fermions by a gauge group
$\mathcal{G}$, which is generated by a subset of Majorana string operators
$\mathcal{G}\subset \text{Maj}(2n)$.
The stabilizer group of $\mathcal{G}$ is given by
$\mathcal{Z}(\mathcal{G})=\mathcal{C}(\mathcal{G})\cap \mathcal{G}$,
where $\mathcal{C}(\mathcal{G})$ is the centralizer group of $\mathcal{G}$.
We define the bare logical operators as 
$\mathcal{L}_{\text{bare}}=\mathcal{C}(\mathcal{G})\setminus \mathcal{G}$,
the gauge operators as 
$\mathcal{L}_g=\mathcal{G}\setminus\mathcal{Z}(\mathcal{G})$,
and the dressed logical operators as 
$\mathcal{L}_{\text{dressed}}=\left\{ L| L= g l,\  g\in\mathcal{L}_g,  l\in\mathcal{L}_{\text{bare}}\right\}$. 
The distance of the dressed logicals is defined as the minimum weight of all dressed logicals~\cite{vuillot2019code}.

In our lattice surgery construction for method 2,
the gauge group is defined as
$\mathcal{G}=\langle \mathcal{S}_{\text{old}}\cup\mathcal{S}_{\text{merged}}\cup(\prod_j\ii\overline{\gamma}_j\overline{\gamma}_j^\prime)\overline{\gamma}_B\rangle$,
where $\mathcal{S}_{\text{old}}$ denotes the union of stabilizer groups of codes $A$ and $B$ before merging,
$\mathcal{S}_{\text{merged}}$ is the stabilizer group of the merged code,
and $\prod_j\ii\overline{\gamma}_j\overline{\gamma}_j^\prime$
is the product of all logical Majorana operators in codes $A$ and $B$ before merging.
When $D_A=D_B$ ($D_A\equiv |\overline{\gamma}_A|$ and $D_B \equiv |\overline{\gamma}_B|$), 
the stabilizer group of the subsystem code is 
$\left\langle \mathcal{S}_{\text{old}}\cap\mathcal{S}_{\text{merged}}\right\rangle$,
the gauge operators are $\langle (\prod_j\ii\overline{\gamma}_j\overline{\gamma}_j^\prime)\overline{\gamma}_B,  c_{B,1},\dots,  c_{B,D_B-1}, M_1,\dots, M_{D_B} \rangle$, 
and the bare logicals consist of all logical operators in codes $A$ and $B$ 
except $\overline{\gamma}_A$ and $\overline{\gamma}_B$. We note that
this construction of gauge group does not apply to the fermionic 
lattice surgery using method 1, as more operators need to be included in the gauge group~\cite{vuillot2019code}.

We now prove that the minimum weight of the dressed logicals is at least $\min\{d_A,d_B\}$,
where $d_A$ and $d_B$ are the code distances of a fermionic LDPC code $A$ and color code $B$, respectively.
Consider an arbitrary dressed logical $ L= g l$ 
with $ g\in\mathcal{L}_g$ and $ l\in\mathcal{L}_{\text{bare}}$,
we analyze its weight $|L|$ according to the support of $l$.

If $\text{Supp}( l)\subset B$, where $\text{Supp}(\dots)$ denotes the set
of physical fermionic sites supported by the operator,
then $ l \sim\overline{\gamma}_B^\prime$. 
Since $( g)_B$ act only on $\gamma$-type Majorana modes,
where $(\dots)_B$ denotes the operator supported on the Majorana modes of $B$,
the weight $| l|$ cannot be reduced below $D_B$,

For the case with $\mathrm{Supp}( l)\subset A$, 
we decompose $g$ as 
$ g= l^\prime m c_B a$, 
where $ m\in \langle M_1,\dots, M_{D_B} \rangle$, 
$ c_B\in\langle  c_{B,1},\dots,  c_{B,D_B-1}\rangle$,
$ l^\prime\in\{\overline{\gamma}_A, I\}$, 
and $ a\in\{(\prod_j\ii\overline{\gamma}_j\overline{\gamma}_j^\prime)\ii\overline{\gamma}_B\overline{\gamma}_A,I\}$
(here if $l^\prime=\overline{\gamma}_A$, then 
$a=(\prod_j\ii\overline{\gamma}_j\overline{\gamma}_j^\prime)\ii\overline{\gamma}_B\overline{\gamma}_A$;
otherwise, both are equal to the identity operator).
Then,
\begin{equation}
\begin{aligned}
| L|&=| l l^\prime m c_B a|\\
&\geq|( l l^\prime m c_B a)_A|+|( l l^\prime m c_B a)_B|\\
&=|( l l^\prime m a)_A|+|( m c_B)_B|+|( a)_B|\\
&\geq|( l l^\prime m a)_A|+|( m)_B|+|( a)_B|\\
&=|( l l^\prime m a)_{A\cup B}|\\
&\geq| l l^\prime a| \\
&\geq\min\{d_A, d_B\}.
\end{aligned}
\end{equation}
The third line holds because $ l l^\prime$ is supported on $A$, 
$( m c_B)_B$ is supported on $\gamma$-type modes, 
and $( a)_B$ is supported on $\gamma^\prime$-type modes.
For the fourth line, if $( m)_B\sim \overline{\gamma}_B$,
$|(m)_B c_B| \ge |(m)_B|$ since $|(m)_B|$ is the distance of the color code.
Otherwise, since $\text{Supp}(( m)_B)\subseteq \text{Supp}(\overline{\gamma}_B)$,
we write $(m)_B=\overline{\gamma}_B o_B$ where 
$\text{Supp}(o_B) \subseteq \text{Supp}(\overline{\gamma}_B)$ and
$\text{Supp}(o_B) \cap \text{Supp}((m)_B)=\emptyset$.
Thus, $c_B ( m)_B=c_B \overline{\gamma}_B o_B=\overline{\gamma}_{B,2} o_B$,
where $\overline{\gamma}_{B,2}=c_B \overline{\gamma}_B$ is still a logical operator of the color code, whose 
weight cannot be smaller than $|\overline{\gamma}_B|$. Let us further write 
$o_B\sim \gamma_{B,i_1}\dots \gamma_{B,i_r} \gamma_{B,i_{r+1}} \dots \gamma_{B,i_{j}}$
where $0\le r \le j \le D_B$, 
$i_1,\dots,i_r \in \text{Supp}(\overline{\gamma}_{B,2})$, and
$i_{r+1},\dots,i_{j} \not\in \text{Supp}(\overline{\gamma}_{B,2})$.
Thus,
$|\overline{\gamma}_{B,2} o_B|=|\overline{\gamma}_{B,2}|+j-2r 
\ge |\overline{\gamma}_B|+j-2r=|(m)_B|+2(j-r) \ge |(m)_B|$.
For the sixth line, if $a l^\prime \sim I$, then 
$|( l l^\prime m a)_{A\cup B}|=|( l m )_{A\cup B}|$. We decompose $l$ as 
$l=l_0 \gamma_{A,j_1}\dots \gamma_{A,j_s}$ where $l_0$ does not contain any operator
in $\{\gamma_{A,1},\dots,\gamma_{A,D_A}\}$.
We also have
$(m)_{A\cup B}\sim \gamma_{A,i_1}\gamma_{B,i_1}\dots \gamma_{A,i_p}\gamma_{B,i_p}$,
where $0\le p \le D_B$ and $ 1 \le i_1,\dots,i_p \le D_B$.
Thus, multiplying $m$ to $l$ only changes $\gamma_{A,i}$ in $l$ to $\gamma_{B,i}$ if 
it is present in $m$. Thus, $|( l m )_{A\cup B}|\ge |l| \ge \min\{d_A,d_B\}$.
Similarly, when $l^\prime a \sim (\prod_j \overline{\gamma}_j\overline{\gamma}_j^\prime)\overline{\gamma}_B$,
we have $|( l l^\prime m a)_{A\cup B}| \ge |l l^\prime a|$. It follows that 
$|l l^\prime a| \ge |(a)_B|\ge \min\{d_A,d_B\}$.

If $l$ contains operators in both codes $A$ and $B$, then $l\sim \overline{\gamma}_B^\prime l_A$
where $\mathrm{Supp}(l_A)\subset A$. Then,
$|L|=|gl|=|\overline{\gamma}_B^\prime|+|gl_A| \ge |\overline{\gamma}_B^\prime|+\min\{d_A,d_B\}$.

For method 1 for fermionic lattice surgery,
we numerically calculate the distance of the dressed logical operators,
and the results are summarized in the table below. We find that the dressed 
logical distance is consistently no smaller than the original logical distance.

\begin{table}[htbp]
	\centering
	\begin{tabular}{|c|c|c|c|c|c|c|}
		\hline
		Index & 1 & 2 & 3 & 4 & 5 & 6\\
		\hline
		Code A & Steane code & $d = 5$-Color code & $\text{PG}(2,4)$ & $\text{PG}(2,8)$ & $\text{DB}(5)$ & $\text{DB}(11)$ \\
		Code B & $d = 5$-Color code & $d = 5$-Color code & $d = 5$-Color code & $d = 5$-Color code & $d = 5$-Color code & $d = 7$-Color code\\
		Distance & $3$ & $5$ & $3$ & $5$ & $4$ & $7$ \\
		\hline
	\end{tabular}
	\caption{Distances of subsystem codes during fermionic lattice surgery 
		using method 1. $\text{DB}$ and $\text{PG} $ denote the double-chain 
		bicycle code and projective geometry code, respectively. $\text{DB}(5)$ 
		is the double-chain bicycle code with the code parameter $[[20,4,4]]_\text{f}$ 
		generated by polynomials $1+x^3$ and $x^2+x^3$~\cite{liu2026stacked}. 
		$\text{DB}(11)$ is the double-chain bicycle code with the 
		code parameter $[[44,4,7]]_\text{f}$ generated by polynomials 
		$1+x^{10}$ and $x^2+x^5$~\cite{liu2026stacked}.}
\end{table}

\section{S-4. Details of numerical simulations}

In this section, we will provide more details regarding calculations of logical failure rate  of memory simulation, resource estimation, lattice-surgery simulation and fermionic circuit simulation discussed in the main text.

\subsection{A. Noisy fermion memory}
\subsubsection{Error models}

\begin{figure}[t]
	\centering
	\includegraphics[width=\linewidth]{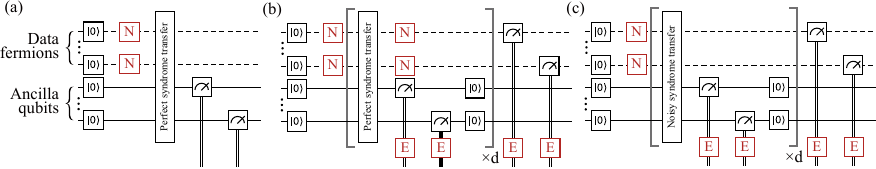}
	\caption{Illustration of three error models: (a) code-capacity, 
		(b) phenomenological, and (c) circuit-level error models.
		Here, the syndrome extraction circuit corresponds to 
		the syndrome extraction transfer used to establish entanglement 
		between data fermions and ancilla qubits, followed by measurements 
		on the ancilla qubits.
		In (a) and (b), the syndrome extraction transfer is noiseless, while 
		in (c), it is noisy as shown in Fig.~\ref{SEcircuit}.
		In all three cases, the initial states are physical $|0\rangle$ states for fermions.
		In (b) and (c), $d$ rounds of syndrome extractions are applied and,
		a transversal measurement on all 
		physical fermion modes in the particle-number basis is performed in the end.
	}
	\label{ErrorModel}
\end{figure}

\begin{figure}[t]
	\centering
	\includegraphics[width=0.7\linewidth]{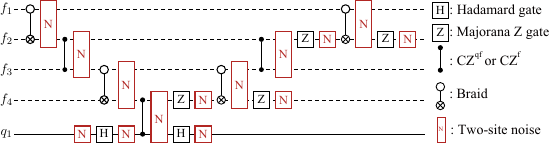}
	\caption{Noisy fermionic syndrome extraction 
		circuit for measuring the stabilizer $-\gamma_1\gamma_2\gamma_3\gamma_4$ 
		at four fermionic sites assisted by a qubit~\cite{Schuckert2024arxiv}.
	}
	\label{SEcircuit}
\end{figure}

To evaluate the performance of the fermion-to-fermion LDPC code, we 
consider three types of error models, similar to the qubit case: code-capacity,
phenomenological and circuit-level error models, as shown in Fig.~\ref{ErrorModel}~\cite{vasic2025quantum}. 

\begin{figure}[t]
	\centering
	\includegraphics[width=0.5\linewidth]{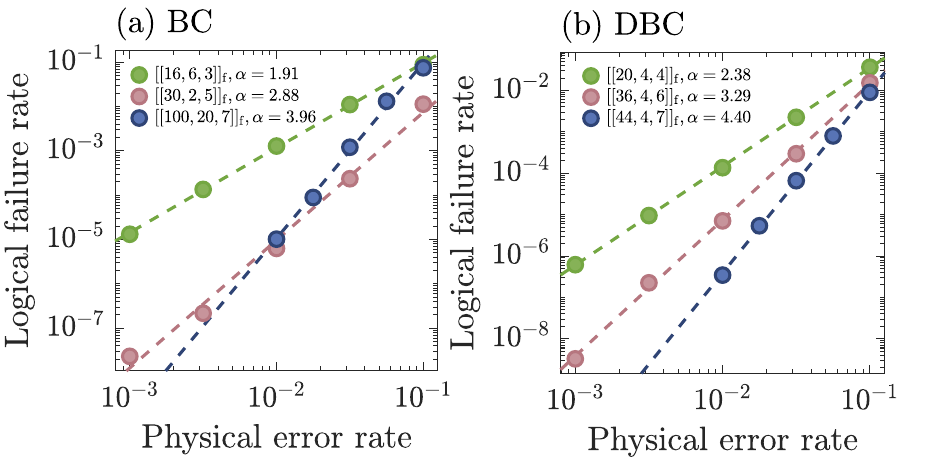}
	\caption{{Logical failure rate with respect to physical error rate for the fermionic $\text{PG}(2,4)$ 
		($[[16,6,3]]_{\mathrm{f}}$) code, the fermionic
		$\text{EG}(2,4)$ ($[[30,2,5]]_{\mathrm{f}}$) code, the fermionic $[[100,20,7]]_{\mathrm{f}}$ bicycle code in (a), and 
		for three double-chain bicycle codes in (b), under code-capacity noise model. Here,
		the BP+OSD decoder is used. 
		The error bars are hidden behind the data points. The dash-dotted lines---similarly 
		depicted in Fig.~\ref{FigPhen}, Fig.~\ref{FigCircuit}, and Fig.~\ref{StackedCodes}---are the linear
		fit of the four smallest data points.}
	}
	\label{FigCap}
\end{figure}

The code-capacity error model considers a single round of errors on fermion sites 
after initialization with noiseless stabilizer measurements, as illustrated in Fig.~\ref{ErrorModel}(a). 
This model represents the upper bound of code performance concerning these errors.
For these errors,
similar to the commonly used one-qubit depolarizing noise channel in qubits~\cite{fowler2009high},
we consider its fermionic version (Majorana depolarizing noise channel) at each fermionic site, 
described by the following Kraus operators for a system 
with $2n$ Majorana modes:
 \begin{equation}\label{SingleBodyFermionicDepolarization}
 	\mathcal{E}_{\text{Maj}}^{(1)} = \left\{\sqrt{1-p} I_j, \sqrt{\frac{p}{3}}\gamma_j, \sqrt{\frac{p}{3}}\gamma'_j, 
 	\sqrt{\frac{p}{3}}\ii \gamma_j \gamma'_j \right\},
 \end{equation}
where $p/3$ denotes the probability of occurrence of either $\gamma_j$, $\gamma'_j$, 
or $\gamma_j \gamma'_j$ at fermionic site $j$.
Figure~\ref{FigCap} shows the logical failure rate with respect to the physical error rate 
for the finite projective geometry $\text{PG}(2,4)$ code, the finite Euclidean geometry 
$\text{EG}(2,4)$ code, the $[[100,20,7]]_{\mathrm{f}}$ bicycle code, and 
three double-chain bicycle codes. 
We clearly see that the logical failure rate is significantly suppressed compared to 
the physical one according to the power law 
$p_{\mathrm{L}} \sim p^\alpha$ with $\alpha \approx \lceil \frac{d}{2} \rceil$.  

\begin{figure}[t]
	\centering
	\includegraphics[width=\linewidth]{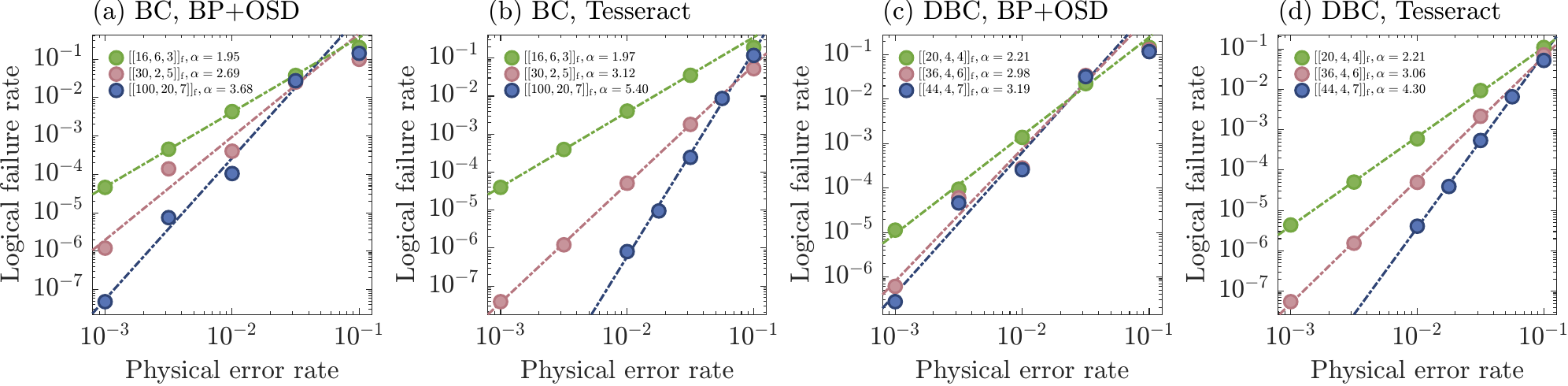}
	\caption{{Logical failure rate as a function of physical error rate for the fermionic $\text{PG}(2,4)$ ($[[16,6,3]]_{\mathrm{f}}$) code, the 
		fermionic $\text{EG}(2,4)$ ($[[30,2,5]]_{\mathrm{f}}$) code, the fermionic $[[100,20,7]]_{\mathrm{f}}$ bicycle code in (a) and (b), and 
		for three double-chain bicycle codes in (c) and (d), under phenomenological noise model.
		In (a) and (c), the BP+OSD decoder is used, whereas in (b) and (d), 
		the Tesseract decoder is employed.
		The error bars are hidden behind the data points. }
	}
	\label{FigPhen}
\end{figure}

Compared with the code-capacity error model, the phenomenological error model
accounts for syndrome extraction errors, as shown in Fig.~\ref{ErrorModel}(b).
We consider $d$ rounds of syndrome extractions, with each extraction followed by a round of
single-fermion-site depolarizing noise channels.
As the phenomenological error model does not include 
the specifics of the syndrome extraction circuit, it does not involve 
error propagation through entangling gates.
Since decoders have a significant impact on the codes' performance, we 
here consider two decoders: the belief propagation and ordered-statistical decoding (BP+OSD) 
algorithm~\cite{roffle2022ldpc} and the 
Tesseract decoder~\cite{beni2025tesseractdecoder} (see the following
subsection for detailed discussion on decoders). 
Figure~\ref{FigPhen} illustrates the logical failure rate as a function of the 
physical error rate for the phenomenological error model for those codes
using the two decoders. We see that  
the logical failure rate remains significantly suppressed compared to 
the physical one. While the results from both decoders generally
follow the scaling $p_{\mathrm{L}} \sim p^\alpha$,
for most analyzed codes, the fitting error for the former decoder is larger than that
of the latter, and the exponent $\alpha$ is smaller than $\lceil \frac{d}{2} \rceil$
for the former, whereas it is larger than $\lceil \frac{d}{2} \rceil$ for the latter.
We attribute the suboptimal performance of the BP+OSD decoder to the presence 
of numerous short cycles with length $l<6$ in the Tanner graph used for decoding~\cite{roffle2022ldpc}; 
for instance, phase errors $i\gamma \gamma'$ can lead to loops of length $l=4$.

For the circuit-level simulation, we also consider the errors occurring 
in the syndrome extraction transfer circuit, in addition to measurement errors and a round of 
one-fermion-site depolarizing noise channels immediately following 
preparation, as shown in Fig.~\ref{ErrorModel}(c).
In the syndrome extraction transfer circuit,
a noise channel, parameterized by the error rate $p$, is applied immediately after
each operation, including fermion, qubit, fermion-qubit, and reset operations 
(see Fig.~\ref{SEcircuit}).
For one-site operations, we consider the one-qubit or one-fermion-site depolarizing 
noise depending on whether the operation is applied to 
an ancilla qubit or a fermion site.
For two-site operations, we consider the two-site depolarizing noise.
For example, when applied to two fermion sites, 
the two-fermion-site depolarizing noise is used: 
\begin{equation}\label{TwoBodyFermionicDepolarization}
	\mathcal{E}_{\text{Maj}}^{(2)}= \{\sqrt{1-p}I_1I_2,\sqrt{\frac{p}{15}}I_1\gamma_2,\sqrt{\frac{p}{15}}I_1\gamma'_2,\dots,\sqrt{\frac{p}{15}}\gamma_1\gamma'_1\gamma_2\gamma'_2\}.
\end{equation}
If the two-site gate involves a qubit, then
we consider a two-site mixed depolarizing noise obtained by replacing $\gamma_2$, $\gamma'_2$, and $\gamma_2 \gamma'_2$
in the above channel by $X$, $Y$, and $Z$, respectively.

Figure~\ref{FigCircuit} shows the results considering the circuit noise model.
We see that the fitted $\alpha$ is significantly reduced for the $\text{PG}(2,4)$ and 
$\text{EG}(2,4)$ codes compared to the results under the phenomenological error model.
For the $[[100,20,7]]_{\mathrm{f}}$ bicycle code, 
the BP+OSD decoder reduces $\alpha$ to $2.27$; in contrast, 
the Tesseract decoder maintains it at $3.74$. Similarly, for the double-chain bicycle codes, 
the Tesseract decoder demonstrates superior performance, as indicated by 
$\alpha$ close to $\lceil \frac{d}{2} \rceil$
and the pseudo-threshold surpassing $0.1\%$ for the $[[36, 4, 6]]_{\text{f}}$
and $[[44, 4, 7]]_{\text{f}}$ codes.

\begin{figure}[t]
	\centering
	\includegraphics[width=\linewidth]{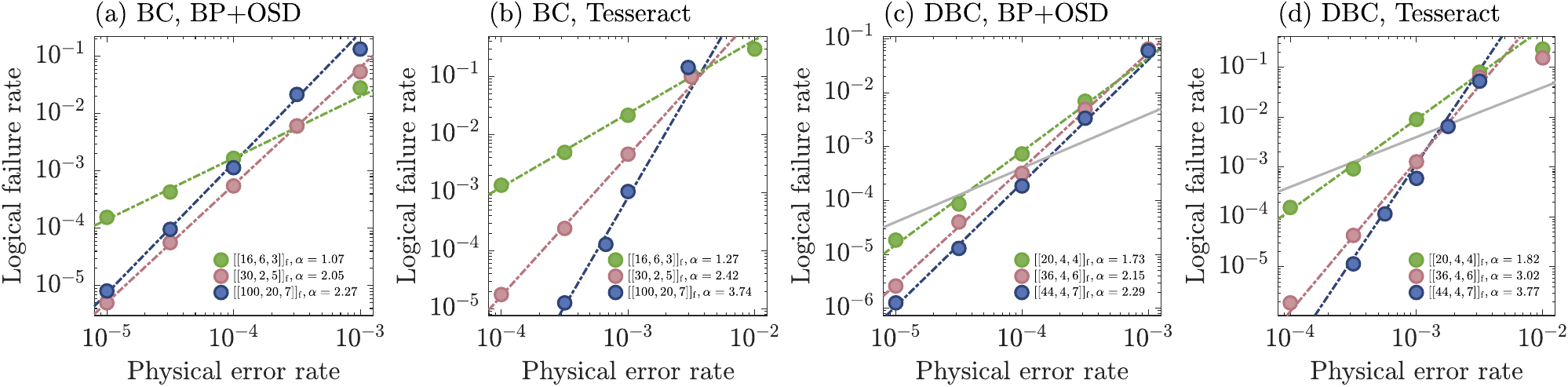}
	\caption{Logical failure rate with respect to physical error rate for the $\text{PG}(2,4)$ ($[[16,6,3]]_{\mathrm{f}}$) code, the 
		$\text{EG}(2,4)$ ($[[30,2,5]]_{\mathrm{f}}$) code, the $[[100,20,7]]_{\mathrm{f}}$ bicycle code in (a) and (b), and 
		for three double-chain bicycle codes in (c) and (d), under circuit noise model.
		In (a) and (c), the BP+OSD decoder is used, whereas in (b) and (d), 
		the Tesseract decoder is employed.
		The error bars are hidden behind the data points. The grey lines in (c) and (d)
		represent the probability that an error occurs on at least one physical site.
		}
\label{FigCircuit}
\end{figure}

In Fig.~\ref{StackedCodes}, we provide the simulation results for 
the fermionic double-layer bivariate bicycle (BB) codes,
double-layer twisted BB codes, and double-layer reflection codes under 
circuit-level noise model using the Tesseract decoder. 
We see that most codes exhibit the pseudo-threshold exceeding $0.1\%$.

\begin{figure}[htbp]
\centering
\includegraphics[width=0.75\textwidth]{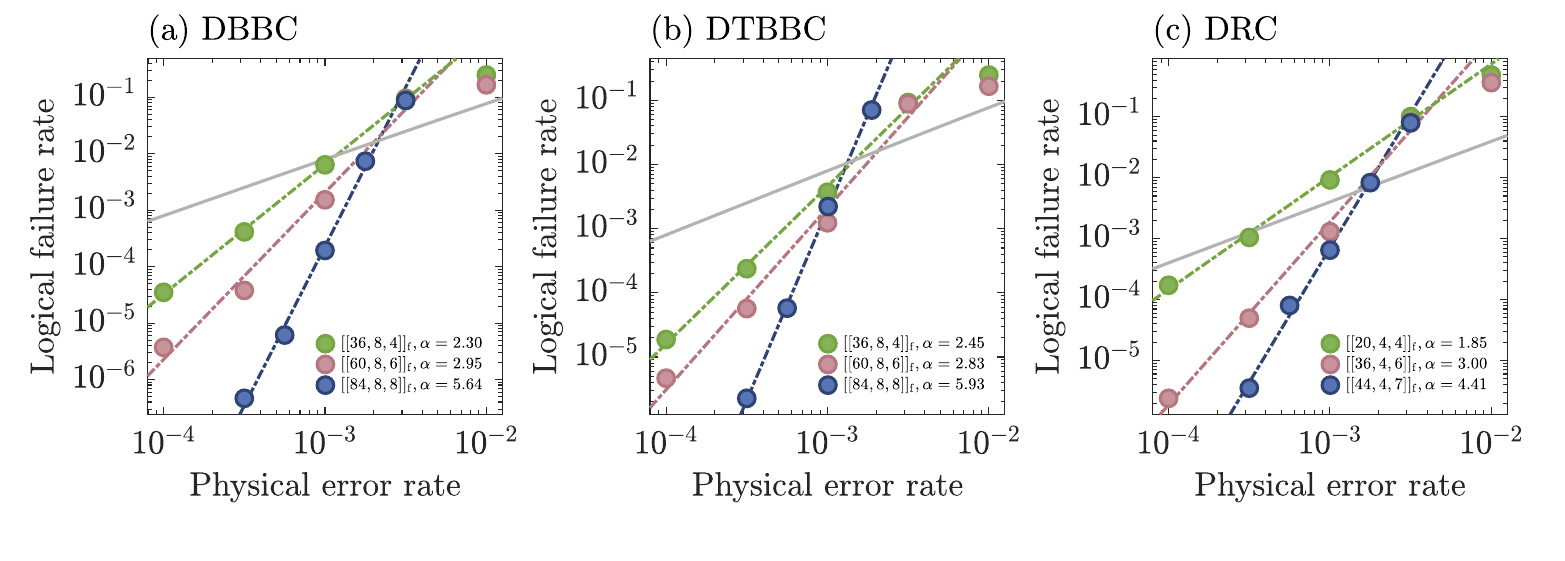}
\caption{ {Logical failure rate versus physical error rate obtained 
under circuit-level error model using the Tesseract decoder for (a) fermionic double-layer BB codes, 
(b) fermionic double-layer twisted BB codes, and (c) fermionic double-layer reflection codes.
The grey lines represent the probability that an error occurs on at least one physical site.} }
\label{StackedCodes}
\end{figure}

We directly simulate the logical failure rate of fermionic teleportation during fermionic lattice surgery, while the fermionic circuit simulation presented in the main text shows that noise can be effectively suppressed during fermionic lattice surgery.

In the simulation, we regard fermionic lattice surgery as a subsystem code.
During the merging and splitting processes, the measurement outcomes of some stabilizers
of the original or merged code are not fixed and therefore cannot be directly used for error correction. 
We measure all stabilizers before and after gauge fixing explicitly, but use only the stabilizers of the corresponding subsystem code that retain the error-correction function
during lattice surgery as detectors.
To ensure that the joint logical observable is measured fault tolerantly, we perform $d$ rounds of syndrome measurements in total, where $d$ is the code distance.
We use the same fermionic circuit-level error model as in the memory
simulations described above.
The resulting detector data are decoded using the BP+OSD decoder.

Figure~\ref{LatticeSurgerySimulation} shows the logical failure rate of
fermionic teleportation as a function of the physical error rate for
lattice surgery. 
The results also exhibit competitive pseudo-thresholds. 
For each code family, the logical failure rate at a fixed physical error rate is suppressed as the code distance increases. 
These numerical results provide direct evidence that the fermionic lattice-surgery protocol is fault tolerant under the circuit-level error model.

\begin{figure}[htbp]
  \centering
  \includegraphics[width=\textwidth]{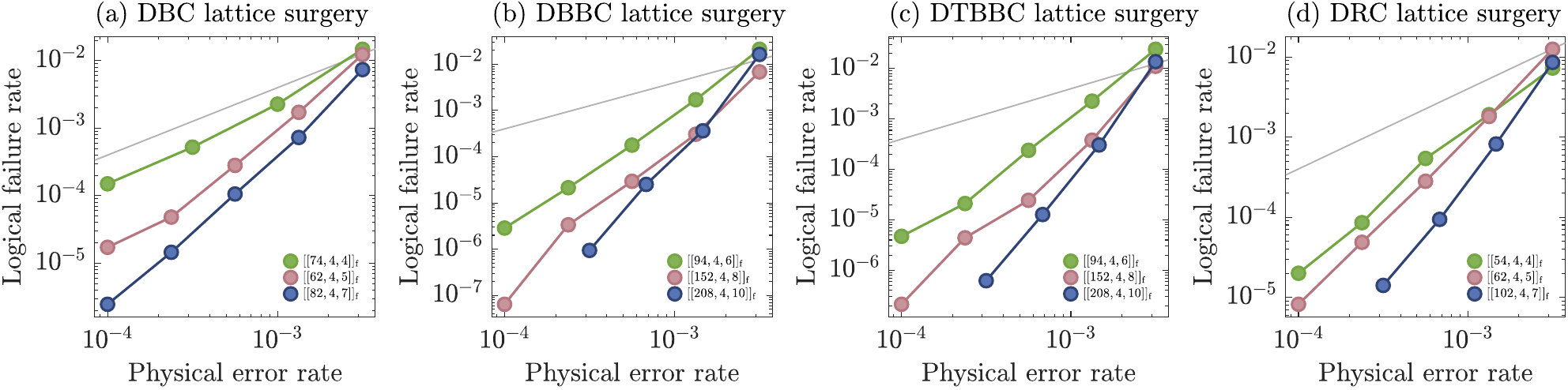}
  \caption{Logical failure rate as a function of the physical error rate
  under the fermionic circuit-level error model using the BP+OSD decoder
  for lattice surgery between fermionic color-code processors and
  (a) fermionic double-chain bicycle codes,
  (b) fermionic double-layer bivariate-bicycle codes,
  (c) fermionic double-layer twisted bivariate-bicycle codes, and
  (d) fermionic double-layer reflection codes.}
  \label{LatticeSurgerySimulation}
\end{figure}

In numerical calculations, similar to Python package stim, when decoders can give the correct 
prediction value for logical operators after performing $N_c=d$ syndrome measurement cycles 
(we set $N_c=d$ for the phenomenological and circuit-level error models), 
it is recorded as a correct event; otherwise, it is recorded as a failure.
The probability of measurement error is on the same order of magnitude as the physical 
error rate; we set the probability of measurement error to $p/3$, where $p$ represents the physical error rate.
The logical error probability is given by $P_L=\frac{N_{\text{failure}}}{N_{\text{sample}}}$,
where $N_{\text{failure}}$ is the number of erroneous trials over total
$N_{\text{sample}}$ trials.
By convention, the logical failure rate is defined as 
$p_L(p)=1 - (1 - P_L)^{1/N_c} \approx P_L/N_c$,
where $P_L$ depends on the physical error rate $p$~\cite{xu2024constant,bravyi2024high}. 
The standard deviation of $P_L$ is $\sigma_{P_L}=\sqrt{\frac{1}{N_{\text{sample}}}P_L(1-P_L)}$
so that the standard deviation of $p_L$ plotted as error bars in the 
numerical simulations of logical failure rate is 
$\sigma_{p_L}=(1/N_c)(1-P_L)^{(1/N_c)-1}\sigma_{P_L}$~\cite{xu2024constant,bravyi2024high}.
We apply a self-adjusting strategy to determine the number of samples $N_{\text{sample}}$ 
for each physical error rate $p$, which ensures a sufficient number of trials with $\sigma_{p_L}$ below $10\%$ of corresponding logical failure rate.
The pseudo-threshold is defined as the physical error rate $p_{\text{th}}$,
such that $p_{\text{L}}(p_{\text{th}}) = P(p_{\text{th}},k)$~\cite{bravyi2024high},
where $P(p,k)=1-(1-p)^k$ represents the probability that at least one 
of $k$ physical fermion sites experiences an error. 

\subsubsection{Detector error model}

With syndromes obtained by stabilizer measurements and transversal measurement
results, we use a decoder to identify the errors that have occurred~\cite{gidney2021stim}.
In addition to the measurement data, it is necessary to provide the decoder
with a detector error model, which explains how errors activate the information 
captured by detectors.
A detector error model is a data structure that encodes the structural 
information of noisy quantum circuits, including the spatiotemporal probability 
distribution of errors and their correlation effects on syndromes and logical 
operators, into a graph representation that can be parsed by decoders~\cite{gidney2021stim}.
The decoder relies on this data format to understand the circuit structure accurately.

To enhance readability, we use the following text format to describe a detector error model:
\begin{equation}
\begin{cases}
\text{Error 1}(p_1)\quad {D}_1, {D}_2, \dots \\
\text{Error 2}(p_2)\quad {D}_2, {L}_1, \dots \\
\dots,
\end{cases}
\end{equation}
where each line corresponds to an independent error event occurring with the probability
specified by the symbol in the bracket, such as $p_1$ or $p_2$. 
Detectors and logical observables triggered by the event are listed on the right-hand side,
where $D_j$ and $L_j$ denote detectors and logical observables, respectively. 
A detector refers to the parity of a set of measurement outcomes in the circuit.
For example,
in Fig.~\ref{DetectorErrorModel}(a), the measurement outcome in the $Z$ basis 
constitutes a detector for a qubit state initialized in $|0\rangle$. 
An $X$ error---an error event---can trigger the detector to yield $-1$. 
Assuming that the $X$ error is the sole trigger, if the detector reads $-1$, 
the error's occurrence is confirmed. This detector cannot identify a $Z$ error,
as it does not trigger it (see Figure~\ref{DetectorErrorModel}(b)). 
In general circuits, attention is required only if measurement outcomes expected to 
be consistent show inconsistency.

\begin{figure}[t]
\centering
\includegraphics[width=0.8\linewidth]{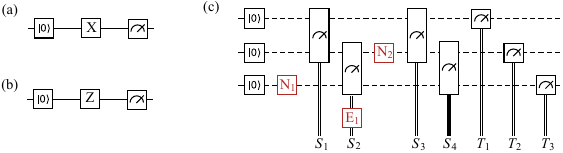}
\caption{(a) For the circuit, a detector is a 
	measurement outcome in the $Z$ basis, and an $X$ error can trigger
	the detector, i.e., flip the value of the detector to $-1$.
	(b) A $Z$ error cannot trigger the detector.
	(c) We use a Majorana repetition code consisting of three fermion 
	sites for demonstration of a detector error model.
	Two noises $N_1$ and $N_2$ and one measurement error $E_1$ 
	are considered. $S_1$ and $S_3$ ($S_2$ and $S_4$) are the measurement outcomes
	of the stabilizer $-\gamma_1\gamma_1'\gamma_2\gamma_2'$ ($-\gamma_2\gamma_2'\gamma_3\gamma_3'$).
	$T_j$ with $j=1,2,3$ denotes the measurement outcome of the particle number
	operator at the $j$th physical fermion site.
	}
\label{DetectorErrorModel}
\end{figure}

Next, we use a test circuit of a Majorana repetition code consisting of three fermion sites as 
an example to demonstrate how to derive the detector error model from a fermionic circuit.
The code's stabilizer generators are $-\gamma_1\gamma_1'\gamma_2\gamma_2'$ 
and $-\gamma_2\gamma_2'\gamma_3\gamma_3'$.
For simplicity, we consider three errors: a $\gamma_3$ error denoted as $N_1$, 
a measurement error denoted as $E_1$, and a depolarizing error 
denoted as $N_2$, as illustrated in Fig.~\ref{DetectorErrorModel}(c).
Under ideal, noiseless conditions, the two stabilizers always remain 
at $+1$, and the values of consecutive syndrome measurements do not vary.
However, an error can cause a change detectable by detectors.
We therefore consider the following detectors: 
$D_1=S_1$, $D_2=S_1 S_3$, $D_3=S_3 T_1 T_2$, $D_4=S_2$,
$D_5=S_2 S_4$, and $D_6=S_4 T_2 T_3$, where $S_j$ with $j=1,2,3,4$ 
denote the syndrome extraction values and $T_j$ with $j=1,2,3$ denote the transversal 
measurement outcomes in the computational basis.
The logical observable is defined by $L_1=T_1 T_2 T_3$,
which should ideally be $+1$ under noiseless conditions.

Deriving the detector error model requires us to analyze the 
events triggered by each independent error occurrence.
When $N_1$ occurs, the stabilizers that anticommute with 
$\gamma_3$ flip sign, causing $S_2$ to flip sign, which 
triggers detector $D_4$ and the logical observable of $L_1$.
When $E_1$ occurs, detectors $D_4$ and $D_5$ are triggered.
$N_2$ corresponds to a depolarizing noise described by the Kraus 
operators of $\{\sqrt{1-p} I,\sqrt\frac{p}{3} \gamma_2,\sqrt\frac{p}{3} \gamma'_2,\sqrt\frac{p}{3}\ii \gamma_2\gamma'_2\}$.
Such an error cannot be described by a single independent event, since
the occurrence of the error of $\gamma_2$ indicates that other errors do
not happen. We therefore need to express this noise channel in terms of
three independent error channels: $\{\sqrt{1-p'} I,\sqrt{p'}\gamma_2\},\{\sqrt{1-p'} I,\sqrt{p'}\gamma'_2\},\{\sqrt{1-p'} I,\sqrt{p'}\ii\gamma_2\gamma'_2\}$.
In the Kraus-operator formulation, it can be shown that the two are equivalent when $p'=\frac{1}{2}(1-\sqrt{1-\frac{4}{3}p})$.
Thus, we decompose the channel of $N_2$ into three independent channels, 
representing errors of $\gamma_2$, $\gamma'_2$, and $\ii \gamma_2 \gamma_2'$, 
respectively.
Similarly, a two-site depolarizing noise can be decomposed into 
$15$ independent two-term Kraus channels with each error occurring with 
the probability of $\frac{1}{2} \left( 1 - \left(1 - \frac{16}{15}p\right)^{1/8} \right)\approx\frac{p}{15}$.
The error of either $\gamma_2$ or $\gamma'_2$ triggers 
$D_2$, $D_5$, and $L_1$, while the error of $\ii \gamma_2 \gamma_2'$
triggers no events and thus can be ignored.
We denote the occurrence probabilities of the four errors including 
$N_1$, $E_1$ and the two independent errors caused by $N_2$
as $p_1$, $p_2$, $p_3$, and $p_4$, respectively (note that $p_3=p_4$).
The resulting detector error model is as follows:
\begin{equation}
\begin{cases}
\text{Error 1}(p_1)\quad\text{D}_4,\text{L}_1\\
\text{Error 2}(p_2)\quad\text{D}_4,\text{D}_5\\
\text{Error 3}(p_3)\quad\text{D}_2,\text{D}_5,\text{L}_1\\
\text{Error 4}(p_3)\quad\text{D}_2,\text{D}_5,\text{L}_1.
\end{cases}
\end{equation}

To efficiently compute the detector error model for fermionic circuits, 
we extend the concept of the Pauli frame to the Majorana frame. 
Currently, most decoders can decode based on the detector error model 
without distinguishing the specific physical implementation of the underlying circuit. 
This enables the use of decoders for fermionic circuits similarly to their use 
for qubit circuits. We have incorporated the numerical method for 
generating detector error models for fermionic circuits into the Python package~\cite{gidney2021stim}.

\subsubsection{Decoder setups}

In this work, we use two independent decoders: 
the BP+OSD algorithm~\cite{roffle2022ldpc},
a most commonly used decoder for quantum LDPC codes, and 
the Tesseract decoder based on $A^*$ search and beam search strategies~\cite{beni2025tesseractdecoder}.

For the former, we employ its Python package stimbposd implementation~\cite{roffle2022ldpc}.
The core workflow of the BP-OSD algorithm is summarized as follows:
It first converts a detector error model into
a Tanner graph, an undirected graph, where vertices represent errors, 
detectors, or logical observables and each line in the detector error model 
corresponds to several edges with each edge indicating that an error triggers
a detector or logical observable.
After this, the BP algorithm is used to update the marginal probabilities 
iteratively of each error through $N_{\text{BP}}$ iterations in total.
This BP process provides the likelihood of each error occurring.
Based on the results of the BP process,
the OSD process sorts all errors in descending order of likelihood 
and selects the top $k$ errors as highly probable occurrences.
Furthermore, the algorithm selects multiple flips of some errors among $w$ errors,
 i.e., assuming that they did not occur, and uses Gaussian elimination to solve 
which remaining errors need to occur to satisfy the detector occurrences.
Here, $w$ is called the OSD order.
Finally, it calculates the probabilities of various solutions, select the one with 
the highest probability as the error pattern, and analyze whether logical observables 
flip under this error pattern.
In our work, the configuration parameters for the BP+OSD decoder are set 
as follows: the number of BP iterations $N_{\text{BP}} = 100$, the OSD candidate 
set size $k = 200$, and the OSD search order $w = 8$ (sweeping mode).

The Tesseract decoder is a maximum likelihood estimation (MLE) decoder 
developed by Google Quantum AI~\cite{beni2025tesseractdecoder}, and
we employ its open-source implementation as a benchmark for comparison.
This decoder introduces heuristic guidance, directing the search 
direction through a precomputed heuristic function, prioritizing 
exploration of high-probability error patterns.
Additionally, the decoder employs beam pruning strategies, retaining 
only the top $k_{\text{beam}}$ candidate error patterns at each search level, 
balancing accuracy and efficiency.
The parameters used in this work when calling this decoder include 
beam width $k_{\text{beam}} = 60$.

\subsection{B. Resource estimation for fermionic memories}
\label{sec:resource-estimation}

In Fig.~1 of the main text, we show that fermionic LDPC codes significantly reduce the resource overhead compared with fermionic color codes. 
In this section, we explain how this conclusion is obtained, providing the details of the underlying numerical simulations.

First, we define the benchmark for comparison. 
We consider two fermionic memories, one encoded using multiple fermionic color-code
blocks and the other using multiple fermionic LDPC-code blocks.
Both memories are required to encode $k$ logical fermionic modes and to achieve a target logical failure rate $p_{\mathrm L}^{\mathrm{target}}$. 
In the resource estimation shown in Fig.~1, the physical error rate is fixed at $p=10^{-3}$.

For the fermionic color-code memory, each color-code block encodes one logical fermionic mode.
Therefore, $k$ color-code blocks are required to encode $k$ logical fermionic modes. 
Let $p_{\mathrm L}$ denote the logical failure rate of a single block. 
The combined logical failure rate of the $k$ blocks is $1-\left(1-p_{\mathrm L}\right)^k$, and we require
\begin{equation}
1-\left(1-p_{\mathrm L}\right)^k
\leq p_{\mathrm L}^{\mathrm{target}}.
\label{eq:color-resource-condition}
\end{equation}
The logical failure rate of an individual code block is reduced by increasing its code distance. 
For a triangular fermionic color code of distance $d$, the number of physical fermionic modes in each block is $n_{\mathrm{block}}=\frac{3d^2+1}{4}$.
The smallest code distance $d$ satisfying Eq.~\eqref{eq:color-resource-condition} therefore determines the total number of physical fermionic modes required by the color-code memory: $n_{\mathrm{color}}=k\frac{3d^2+1}{4}$.

For the fermionic LDPC memory, we use fermionic double-layer BB codes. 
We consider each double-layer BB code block encodes $k'=8$ logical fermionic modes, and hence $\frac{k}{k'}$ code blocks are required. 
If $p_{\mathrm L}$ denotes the logical failure rate of a single double-layer BB block, the logical failure rate of the entire memory is $1-\left(1-p_{\mathrm L}\right)^{\frac{k}{k'}}$, which is required to satisfy
\begin{equation}
1-\left(1-p_{\mathrm L}\right)^{\frac{k}{k'}}\leq p_{\mathrm L}^{\mathrm{target}}.
\label{eq:ldpc-resource-condition}
\end{equation}
We then search for double-layer BB codes that retain $k'=8$ encoded logical fermionic modes but have progressively larger code distances, until Eq.~\eqref{eq:ldpc-resource-condition} is satisfied. 

We directly calculate the logical failure rates for color codes and
double-layer BB codes containing fewer than $300$ physical fermionic
modes. 
These numerical results are fitted and used to extrapolate the
logical failure rates to larger code distances. 
For the color-code memory, once the required distance $d$ is obtained from the
extrapolation, the corresponding physical resource follows directly
from the exact relation.
For the double-layer BB codes, an additional extrapolation is required
to convert the code distance into the number of physical fermionic
modes. 
We fit the relation between the number $n'$ of physical
fermionic modes in a single block and its code distance $d$ using a
quadratic function.
The extrapolated value of $n'(d)$ is then substituted to obtain the total resource requirement of the LDPC memory.

\subsection{C. Fermionic circuit simulation}

We now provide the details of Fig.~4 in the main text. 
There, {we choose the fermionic $[[88,4,7]]_{\text{f}}$ double-chain bicycle code as 
the fermionic memory and use the four logical modes for simulation 
and two fermionic $[[37,1,7]]_{\text{f}}$ color codes as the processors.}
The memory is initialized in the $|\overline{0}\overline{1}\overline{0}\overline{1}\rangle$ state.
We then perform multiple rounds of the circuit shown in Fig.~4(a), where each round represents a time step.
To implement the tunneling gate or fSWAP gate between logical modes, we transfer 
the corresponding logical modes from the memory to the processors via fermionic lattice surgery, 
perform the gate on the processors, and then transfer the logical modes back to the memory.
The processors are then reset to $|\overline{0}\overline{0}\rangle$ state, 
preparing them for the next tunneling or fSWAP gate.
We apply the same circuit-level error model as described previously, 
with a physical error rate of $p = 10^{-4}$.
We execute the circuit with and without error correction.
{For simplicity, we assume that onsite measurements in both Fig.~3(a) and Fig.~4(a) are 
noiseless.
For all other cases, we account for circuit-level errors with $d$ cycles of syndrome 
extraction followed by error correction.
The decoder used here is the Tesseract decoder~\cite{beni2025tesseractdecoder}.}
We track the expectation value of logical particle numbers at sites 1 and 4 
at each time step to compare the simulation results of both scenarios.

In the main text, we employ the first method to perform fermionic lattice surgery.
Here, in Fig.~\ref{figs-simulation}, we provide the simulation results achieved 
using the second method for lattice surgery.

\begin{figure}[htbp]
	\centering
	\includegraphics[width=0.35\textwidth]{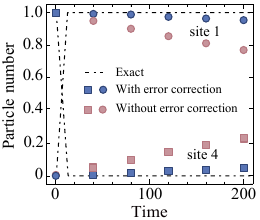}
	\caption{The time evolution of logical particle numbers at sites 1 and 4, which are the same as in Fig.~4(b) except that the lattice surgery is performed using the second method.
	The error bars are hidden behind the symbols.}
	\label{figs-simulation}
\end{figure}
\end{widetext}
\end{document}